\documentclass[aap,MSNbibl,citesort,dvips]{arximspdf}
\usepackage{graphicx}
\doi{10.1214/11-AAP806} 
\volume{23}
\issue{1}
\pubyear{2013}
\firstpage{66}
\lastpage{98}

\makeatletter

\newtheorem{theorem}{Theorem}[section]

\newproclaim{ass}[theorem]{Assumption}
\newproclaim{alg}[theorem]{Algorithm}

\newtheorem{prop}[theorem]{Proposition}
\newtheorem{lemma}[theorem]{Lemma}
\newtheorem{cor}[theorem]{Corollary}

\newproclaim{example}[theorem]{Example}
\newproclaim{defi}[theorem]{Definition}
\newproclaim{remark}[theorem]{Remark}

\newcommand{\lancuch}{(X_n)_{n\geq0}}
\newcommand{\lancucht}{(\tilde{X}_n)_{n\geq0}}
\newcommand{\lancuchs}{(S_n)_{n\geq0}}
\newcommand{\stany}{\mathcal{X}}
\newcommand{\borel}{\mathcal{B}(\stany)}

\newcommand{\dd}{\mathrm{d}}

\newcommand{\Rl}{\mathbb{R}}
\renewcommand{\Pr}{\mathbb{P}}

\newcommand{\Y}{\mathcal{Y}}
\newcommand{\half}{{1 \over2}}

\makeatother

\begin{document}
\begin{frontmatter}

\title{Adaptive Gibbs samplers and related MCMC methods}
\runtitle{Adaptive Gibbs samplers}

\begin{aug}
\author[A]{\fnms{Krzysztof} \snm{{\L}atuszy\'{n}ski}\thanksref{t1,t2,t3}\ead[label=e1]{latuch@gmail.com}},
\author[A]{\fnms{Gareth O.} \snm{Roberts}\corref{}\thanksref{t3}\ead[label=e2]{Gareth.O.Roberts@warwick.ac.uk}}\\
\and
\author[B]{\fnms{Jeffrey S.} \snm{Rosenthal}\thanksref{t2}\ead[label=e3]{jeff@math.toronto.edu}}
\runauthor{K. {\L}atuszy\'nski, G. O. Roberts and J. S. Rosenthal}
\affiliation{University of Warwick, University of Warwick and
University of Toronto}
\address[A]{K. {\L}atuszy\'nski\\
G. O. Roberts\\
Department of Statistics\\
University of Warwick\\
CV4 7AL, Coventry\\
United Kingdom \\
\printead{e1}\\
\hphantom{E-mail: }\printead*{e2}}
\address[B]{J. S. Rosenthal\\
Department of Statistics\\
University of Toronto\\
Toronto, Ontario M5S 3G3\\
Canada\\
\printead{e3}} 
\end{aug}

\thankstext{t1}{While preparing this manuscript K. {\L}atuszy\'nski was
a postdoctoral
fellow at the Department of Statistics, University of Toronto.}

\thankstext{t2}{Supported in part by NSERC of Canada.}

\thankstext{t3}{Supported in part by CRiSM and other grants from EPSRC.}

\received{\smonth{1} \syear{2011}}
\revised{\smonth{6} \syear{2011}}

%
\begin{abstract}
We consider various versions of adaptive Gibbs and
Metropolis-within-Gibbs samplers, which update their
selection probabilities (and perhaps also their proposal
distributions) on the fly during a run by learning as
they go in an attempt to optimize the algorithm. We
present a cautionary example of how even a simple-seeming
adaptive Gibbs sampler may fail to converge. We then
present various positive results guaranteeing convergence
of adaptive Gibbs samplers under certain conditions.
\end{abstract}

%

\end{frontmatter}

\section{Introduction} \label{secintro}

Markov chain Monte Carlo (MCMC) is a commonly used approach to evaluating
expectations of the form $\theta:=\int_{\stany}f(x)\pi(\dd x)$, where
$\pi$
is an intractable probability measure, for example, known up to a normalizing
constant. One simulates $\lancuch$, an ergodic Markov chain on $\stany$,
evolving according to a transition kernel $P$ with stationary limiting
distribution $\pi$ and, typically, takes ergodic average as an estimate
of $\theta$. The approach is justified by asymptotic Markov chain theory
(see, e.g.,~\cite{MeynTw,RobRos}). Metropolis algorithms and Gibbs samplers
(to be described in Section~\ref{secsettingcounter}) are among the
most common MCMC algorithms; cf. \mbox{\cite{RoCa,Liu,RobRos}}.

The quality of an estimate produced by an MCMC algorithm depends on
probabilistic properties of the underlying Markov chain. Designing an
appropriate transition kernel $P$ that guarantees rapid convergence
to stationarity and efficient simulation is often a challenging task,
especially in high dimensions. For Metropolis algorithms there are
various optimal scaling results \cite
{RobertsGelmanGilks,RobRosMALA,Bedardaap,Bedardbeyond,AtchRobRosMCMCMC,RobRosscaling2001,RobRos,Rosenthalproposal}
which provide ``prescriptions'' of how to do
this, though they typically depend on unknown characteristics of $\pi$.

For random scan Gibbs and Metropolis-within-Gibbs samplers, a further
design decision is choosing the
selection probabilities (i.e., coordinate weightings) which will be used
to select which coordinate to update next. These are usually chosen to be
uniform, but some recent work \cite
{LiuWongKong,Levine05a,Levine05b,DiaconisStatSci,sylvia1,sylvia2} has
suggested that nonuniform
weightings may sometimes be preferable.


For a very simple toy example to illustrate this issue, suppose $\stany
= [0,1] \times[-100,100]$, with $\pi(x_1,x_2) \propto x_1^{100}
(1+\sin(x_2))$. Then with respect to $x_1$, this $\pi$ puts almost all
of the mass right up against the line $x_1=1$. Thus, repeated Gibbs
sampler updates of the coordinate $x_1$ provide virtually no help in
exploring the state space, and do not need to be done often at all
(unless the functional $f$ of interest is \textit{extremely} sensitive to
tiny changes in $x_1$). By contrast, with respect to $x_2$, this $\pi$
is a highly multi-modal density with wide support and many peaks and
valleys, requiring many updates to the coordinate $x_2$ in order to
explore the state space appropriately. (Of course, as with any Gibbs
sampler, repeatedly updating one coordinate does not help with
\textit{distributional} convergence; it only helps with sampling the
entire state
space to produce good estimates.) Thus, an efficient Gibbs sampler for
this example would not update each of $x_1$ and $x_2$ equally often;
rather, it would update $x_2$ very often and $x_1$ hardly at all.
Of course, in this simple example, it is easy to see directly that $x_1$
should be updated less than $x_2$, and furthermore, such efficiencies
would only improve the sampler by approximately a factor of 2. However,
in a high-dimensional example (cf.~\cite{sylvia2}), such issues could
be much more significant, and also much more difficult to detect manually.

One promising avenue to address this challenge
is \textit{adaptive MCMC algorithms}.
As an MCMC simulation progresses, more and more information about the
target distribution $\pi$ is learned. Adaptive MCMC attempts to use
this new information to redesign the transition kernel $P$ on the fly,
based on the current simulation output. That is, the transition kernel
$P_n$ used for obtaining $X_n|X_{n-1}$ may depend on $\{X_0, \ldots,
X_{n-1}\}$.
So, in the above
toy example, a good adaptive Gibbs sampler would somehow automatically
``learn'' to update $x_1$ less often, without requiring the user to
determine this manually (which could be difficult or impossible in a
very high-dimensional problem).

Such adaptive algorithms are only valid if their ergodicity
can be established. Unfortunately the stochastic process $\lancuch$ for
an adaptive
algorithm is no longer a Markov chain; the potential benefit of adaptive
MCMC comes at the price of requiring more sophisticated theoretical
analysis. There is substantial and rapidly growing literature on both
theory and practice of adaptive MCMC (see, e.g., \cite
{Gilksadap,Haario,aro,AndrieuMoulines,Haario2,Kadane,RobRosJAP,RobRosex,LatPhD,chao2,chao3,chao4,BaiRobRos,Bai2,Bai3,SaksmanVihola,Vihola,AtchadeFort,AtchadeEtAl,Craiu2})
which includes counterintuitive
examples where $X_n$ fails to converge to the desired distribution
$\pi$ \mbox{(cf.~\cite{aro,RobRosJAP,BaiRobRos,LatPhD})}, as well as many
results guaranteeing ergodicity under various assumptions. Most of the
previous work on ergodicity of adaptive MCMC\vadjust{\goodbreak} has concentrated on adapting
Metropolis and related algorithms, with less attention paid to ergodicity
when adapting the selection probabilities for random scan Gibbs samplers.


Motivated by such considerations, in the present paper we study the
ergodicity of various types of adaptive Gibbs samplers. To our knowledge,
proofs of ergodicity for adaptively-weighted Gibbs samplers have
previously been considered only by~\cite{LevineCasella}, and we shall
provide a counter-example below (Example~\ref{exstairwaytoheaven2})
to demonstrate that their main result is not correct. In view of this,
we are not aware of any valid ergodicity results in the literature that
consider adapting selection probabilities of random scan Gibbs samplers,
and we attempt to fill that gap herein.

This paper is organized as follows. We begin in
Section~\ref{secsettingcounter} with basic definitions.
In Section~\ref{seccounterex} we present a cautionary
Example~\ref{exstairwaytoheaven2}, where a seemingly ergodic
adaptive Gibbs sampler is in fact transient (as we prove formally later
in Section~\ref{seccounterproof}) and provides a counter-example
to Theorem 2.1 of~\cite{LevineCasella}. Next, we establish various
positive results for ergodicity of adaptive Gibbs samplers.
We consider adaptive random
scan Gibbs samplers (\texttt{AdapRSG}) which update coordinate selection
probabilities as the simulation progresses, adaptive random scan
Metropolis-within-Gibbs samplers
(\texttt{AdapRSMwG}) which update coordinate selection probabilities
as the simulation progresses and adaptive random scan adaptive
Metropolis-within-Gibbs samplers
(\texttt{AdapRSadapMwG}) that update coordinate selection
probabilities as
well as proposal distributions for the Metropolis steps. Positive
results in the uniform setting are discussed in Section \ref
{secerguniform}, whereas Section~\ref{secergnonunif} deals with the
nonuniform setting. In each case, we prove that under
reasonably mild conditions, the adaptive Gibbs samplers are guaranteed
to be ergodic, although our cautionary example does show that it is
important to verify some conditions before applying such algorithms.

\section{Preliminaries}
\label{secsettingcounter}

Gibbs samplers are commonly used MCMC algorithms for sampling from
complicated high-dimensional probability distributions $\pi$ in cases
where the full conditional distributions of $\pi$ are easy to sample from.
To define them,
let $(\stany, \mathcal{B}(\mathcal{X}))$ be a $d$-dimensional state
space where $\stany= \stany_1 \times\cdots\times\stany_d$ and write
$X_n \in\stany$ as $X_n = (X_{n,1}, \ldots, X_{n,d})$. We shall use
the shorthand notation
\[
X_{n,-i}:=
(X_{n,1}, \ldots, X_{n,i-1}, X_{n,i+1},
\ldots, X_{n,d})
\]
and similarly $\stany_{-i} = \stany_1\times\cdots
\times\stany_{i-1}\times\stany_{i+1}\times\cdots\times\stany_d$.

Let
$\pi(\cdot| x_{-i})$ denote the conditional distribution of $Z_{i} |
Z_{-i}=x_{-i}$ where \mbox{$Z \sim\pi$}.
The random scan
Gibbs sampler draws $X_n$ given $X_{n-1}$ (iteratively for
$n=1,2,3,\ldots$) by first choosing one coordinate
at random according to some selection probabilities $\alpha=
(\alpha_1,\ldots, \alpha_d)$ (e.g., uniformly), and then updating that
coordinate by a draw from its conditional distribution.
More precisely, the Gibbs sampler transition kernel $P =
P_{\alpha}$ is the result of performing the following three
steps.\vadjust{\goodbreak}

%
\begin{alg}[{[\texttt{RSG}($\alpha$)]}]\label{algrandomGibbs}
(1)
Choose coordinate $i \in\{1,\ldots, d\}$ according to
selection probabilities $\alpha$, that is, with $\Pr(i=j) =
\alpha_j$.\vspace*{-6pt}

\begin{longlist}[(3)]
\item[(2)] Draw $Y \sim\pi(\cdot| X_{n-1, -i})$.
\item[(3)] Set $X_{n}: = (X_{n-1,1},\ldots, X_{n-1,i-1}, Y,
X_{n-1,i+1},\ldots, X_{n-1,d})$.
\end{longlist}
\end{alg}

Whereas the standard approach is to choose the coordinate $i$ at the
first step uniformly at random, which corresponds to $\alpha=
(1/d,\ldots, 1/d)$, this may be a substantial waste of simulation
effort if $d$
is large and variability of coordinates differs significantly. This
has been discussed theoretically in~\cite{LiuWongKong} and also
observed empirically, for example, in Bayesian variable selection for linear
models in statistical genetics~\cite{sylvia1,sylvia2}.

Throughout the paper we denote the transition kernel of a random scan
Gibbs sampler with selection probabilities $\alpha$ as $P_{\alpha}$ and
the transition kernel of a single Gibbs update of coordinate $i$ is
denoted as $P_i$, hence, $P_{\alpha} = \sum_{i=1}^d \alpha_i P_i$.

We consider a class of adaptive random scan Gibbs samplers where selection
probabilities $\alpha= (\alpha_1,\ldots, \alpha_d)$ are subject to
optimization within some subset $\Y\subseteq[0,1]^d$
of possible choices. Therefore a single step of our generic adaptive algorithm
for drawing $X_{n}$ given the trajectory $X_{n-1},\ldots, X_0$ and
current selection probabilities $\alpha_{n-1} = (\alpha_{n-1,1},\ldots,
\alpha_{n-1,d})$ amounts to the following steps, where $R_n(\cdot)$
is some update rule for $\alpha_n$.
%
%
\begin{alg}[(\texttt{AdapRSG})]\label{algGibbsadap}
(1) Set $\alpha_{n}:=R_n(\alpha_0,\ldots, \alpha_{n-1},X_{n-1},\ldots
,\allowbreak X_0) \in\Y$.\vspace*{-6pt}

\begin{longlist}[(3)]
\item[(2)] Choose coordinate $i \in\{1,\ldots, d\}$ according to
selection probabilities~$\alpha_{n}$.
\item[(3)] Draw $Y \sim\pi(\cdot| X_{n-1,-i})$.
\item[(4)] Set $X_{n}: = (X_{n-1,1},\ldots, X_{n-1,i-1}, Y,
X_{n-1,i+1},\ldots, X_{n-1,d})$.
\end{longlist}
\end{alg}

Algorithm~\ref{algGibbsadap} defines $P_n$, the transition kernel used
at time
$n$, and $\alpha_n$ here plays the role of $\Gamma_n$ in the more
general adaptive setting of, for example,~\cite{RobRosJAP,BaiRobRos}.
Let $\pi_n = \pi_n(x_0, \alpha_0)$ denote the distribution of $X_n$
induced by Algorithm~\ref{algrandomGibbs} or~\ref{algGibbsadap}, given
starting values $x_0$ and $\alpha_0$, that is, for $B \in\mathcal
{B}(\mathcal{X})$,
%
%
\begin{equation}
\pi_n(B) = \pi_n((x_0, \alpha_0),B):= \mathbb{P}(X_n \in
B |
X_0 = x_0, \alpha_0).
\end{equation}
Clearly, if one uses Algorithm~\ref{algrandomGibbs} then $\alpha
_0=\alpha$ remains fixed and\break $\pi_n(x_0, \alpha)(B) = P_{\alpha}^n(x_0,
B)$. By $\|\nu-\mu\|_{\mathrm{TV}}$ denote the total variation distance between
probability measures $\nu$ and $\mu$. Let
%
%
\begin{equation}\label{eqndefofTn}
T(x_0, \alpha_0, n):= \|\pi_n(x_0,
\alpha_0)-\pi\|_{\mathrm{TV}}.
\end{equation}
We call the adaptive Algorithm~\ref{algGibbsadap} \textit{ergodic} if
$T(x_0, \alpha_0, n) \to0$ for $\pi$-almost every starting state
$x_0$ and all $\alpha_0 \in\Y$.



We shall also consider random scan Metropolis-within-Gibbs
samplers that instead of sampling from the full conditional at step (2)
of Algorithm~\ref{algrandomGibbs} [resp., at step (3) of
Algorithm~\ref{algGibbsadap}], perform a single Metropolis or
Metropolis--Hastings
step~\cite{Metropolis,Hastings}. More precisely, given $X_{n-1,-i}$,
the $i$th coordinate
$X_{n-1, i}$ is updated by a draw $Y$ from the proposal distribution
$Q_{X_{n-1,-i}}(X_{n-1, i}, \cdot)$ with the usual
Metropolis acceptance probability for the marginal stationary
distribution $\pi(\cdot| X_{n-1, -i})$. Such Metropolis-within-Gibbs
algorithms were originally proposed by~\cite{Metropolis} and have been
very widely used. Versions of this algorithm which adapt the proposal
distributions $Q_{X_{n-1,-i}}(X_{n-1, i}, \cdot)$
were considered by, for example,~\cite{Haario2,RobRosex}, but always with
fixed (usually uniform) coordinate selection probabilities.
If instead the proposal
distributions $Q_{X_{n-1,-i}}(X_{n-1, i}, \cdot)$ remain fixed,
but the selection
probabilities $\alpha_i$ are adapted on the fly, we obtain the following
algorithm
[where $q_{x,-i}(x,y)$ is the density function for
$Q_{x,-i}(x,\cdot)$].
%
%
\begin{alg}[(\texttt{AdapRSMwG})]
(1) Set $\alpha_{n}:=R_n(\alpha_0,\ldots, \alpha_{n-1},X_{n-1},\ldots
,\allowbreak X_0) \in\Y$.\vspace*{-6pt}

\begin{longlist}[(3)]
\item[(2)] Choose coordinate $i \in\{1,\ldots, d\}$ according to
selection probabilities~$\alpha_n$.
\item[(3)] Draw $Y \sim Q_{X_{n-1,-i}}(X_{n-1, i}, \cdot)$.
\item[(4)] With probability
%
%
\begin{equation}
\label{acceptprob}
\min\biggl(1, {\pi(Y| X_{n-1, -i})
q_{X_{n-1, -i}}(Y,X_{n-1,i})
\over
\pi(X_{n-1}| X_{n-1, -i})
q_{X_{n-1, -i}}(X_{n-1,i},Y)
} \biggr),
\end{equation}
accept the proposal and set
\[
X_{n} = (X_{n-1,1},\ldots, X_{n-1,i-1}, Y, X_{n-1,i+1},\ldots, X_{n-1,d});
\]
otherwise, reject the proposal and set $X_n = X_{n-1}$.
\end{longlist}
\end{alg}

Ergodicity of \texttt{AdapRSMwG} is considered in Sections
\ref{secadapMwG} and~\ref{secnonunifadapRSadapMwG} below. Of course,
if the proposal distribution $Q_{X_{n-1,-i}}(X_{n-1, i}, \cdot)$ is
symmetric about $X_{n-1}$, then the $q$ factors in the acceptance
probability (\ref{acceptprob}) cancel out, and (\ref{acceptprob})
reduces to the simpler probability $\min(1, \pi(Y| X_{n-1, -i}) /
\pi(X_{n-1}| X_{n-1, -i}))$.

We shall also consider versions of the algorithm in which the proposal
distributions $Q_{X_{n-1,-i}}(X_{n-1, i}, \cdot)$ are also chosen
adaptively, from some family $\{Q_{x_{-i},\gamma}\}_{\gamma
\in\Gamma_i}$ with corresponding density functions $q_{x_{-i},\gamma}$,
as in, for example, the statistical genetics
application~\cite{sylvia1,sylvia2}. Versions of such algorithms with fixed
selection probabilities are considered by, for example,~\cite{Haario2}
and~\cite{RobRosex}. They require additional adaptation parameters
$\gamma_{n,i}$ that are updated on the fly and are allowed to
depend on the past trajectories. More precisely, if $\gamma_n =
(\gamma_{n, 1},\ldots, \gamma_{n, d})$ and $\mathcal{G}_n = \sigma\{
X_0,\ldots, X_n, \alpha_0,\ldots, \alpha_n, \gamma_0,\ldots, \gamma
_n\}$,
then the conditional distribution of $\gamma_{n}$ given $\mathcal{G}_{n-1}$
can be specified by the particular algorithm used, via a second update
function $R'_n$. If we combine such
proposal distribution adaptions with coordinate selection probability
adaptions, this results in a doubly-adaptive algorithm, as follows.
%
%
\begin{alg}[(\texttt{AdapRSadapMwG})]
(1) Set $\alpha_{n}:=R_n(\alpha_0,\ldots, \alpha_{n-1},\break
X_{n-1},\ldots , X_0, \gamma_{n-1},\ldots,\gamma_0)
\in\Y$.\vspace*{-6pt}

\begin{longlist}[(3)]
\item[(2)] Set $\gamma_{n}:=R'_n(\alpha_0,\ldots, \alpha_{n-1},X_{n-1},\ldots
, X_0,
\gamma_{n-1},\ldots,\gamma_0) \in\Gamma_1 \times\cdots\times
\Gamma_n$.
\item[(3)] Choose coordinate $i \in\{1,\ldots, d\}$ according to
selection probabilities $\alpha$, that is, with $\Pr(i=j) = \alpha_j$.
%
\item[(4)] Draw $Y \sim Q_{X_{n-1,-i},\gamma_{n-1, i}}(X_{n-1, i},
\cdot)$.
\item[(5)] With probability given by (\ref{acceptprob}),
\[
\min\biggl(1, {\pi(Y| X_{n-1, -i})
q_{X_{n-1, -i}, \gamma_{n-1, i}}(Y,X_{n-1,i})
\over
\pi(X_{n-1}| X_{n-1, -i})
q_{X_{n-1, -i}, \gamma_{n-1, i}}(X_{n-1,i},Y)
} \biggr),
\]
accept the proposal and set
\[
X_{n} = (X_{n-1,1},\ldots, X_{n-1,i-1}, Y, X_{n-1,i+1},\ldots, X_{n-1,d});
\]
otherwise, reject the proposal and set $X_n = X_{n-1}$.
\end{longlist}
\end{alg}

Ergodicity of \texttt{AdapRSadapMwG} is considered in Sections
\ref{secadaptadapt} and~\ref{secnonunifadapRSadapMwG} below.


\section{A counter-example}
\label{seccounterex}

Adaptive algorithms destroy the Markovian nature of $\lancuch$, and are
thus notoriously difficult to analyze theoretically. In particular,
it is easy to be tricked into thinking that a simple adaptive algorithm
``must'' be ergodic when in fact it is not.

For example, Theorem 2.1 of~\cite{LevineCasella} states that ergodicity
of adaptive Gibbs samplers follows from the following two simple
conditions:

\begin{longlist}[(ii)]
\item[(i)] $\alpha_n \to\alpha$ a.s. for some fixed $\alpha\in
(0,1)^d$; and

\item[(ii)] the random scan Gibbs sampler with fixed selection
probabilities $\alpha$ induces an ergodic Markov chain with stationary
distribution $\pi$.
\end{longlist}


Unfortunately, this claim is false, that is, (i) and (ii) alone
do not guarantee ergodicity, as the following example and
proposition demonstrate. (It seems that in the proof of Theorem 2.1
in~\cite{LevineCasella}, the same measure is used to represent
trajectories of the adaptive process and of a corresponding nonadaptive
process, which is not correct and thus leads to the error.)


%
%
\begin{example} \label{exstairwaytoheaven2}
Let\vspace*{1pt} $\mathbb{N} = \{1,2,\ldots\}$, and let the state space
$\stany= \{(i,j) \in\mathbb{N} \times\mathbb{N}\dvtx i = j$
or $i = j+1\}$, with target distribution given by
$\pi(i,j) \propto j^{-2}$.
On $\stany$, consider a class of adaptive random scan Gibbs samplers
for $\pi$,
as defined by Algorithm~\ref{algGibbsadap}, with update rule given by
%
%
\begin{equation}\label{eqnformulaforalpha}
R_n\bigl(\alpha_{n-1},X_{n-1}=(i,j)\bigr) = \cases{
\displaystyle
\biggl\{\frac{1}{2} + \frac{4}{a_n},\frac{1}{2} -
\frac
{4}{a_n}\biggr\} ,&\quad if $i=j$,\vspace*{2pt}\cr
\displaystyle
\biggl\{\frac{1}{2} - \frac
{4}{a_n},\frac{1}{2} + \frac{4}{a_n}\biggr\} ,&\quad if $i=j+1$}
\end{equation}
for some choice of the sequence $(a_n)_{n=0}^\infty$ satisfying
$8 <a_n \nearrow\infty$.
\end{example}


Example~\ref{exstairwaytoheaven2} satisfies assumptions (i) and (ii)
above. Indeed, (i) clearly holds since $\alpha_n \to\alpha:=
(\half,\half)$, and (ii)\vspace*{1pt} follows immediately from the
standard Markov chain properties of irreducibility and aperiodicity;
cf. \cite {MeynTw,RobRos}. However, if $a_n$ increases to $\infty$
slowly enough, then the example exhibits transient behavior and is not
ergodic. More precisely, we shall prove the following proposition.


%
%
\begin{prop}\label{factXnonergodic} There exists a choice of the
$(a_n)$ for which
the process $\lancuch$ defined in Example~\ref{exstairwaytoheaven2} is
not ergodic. Specifically,
starting at $X_0=(1,1)$, we have
$\mathbb{P}(X_{n,1} \to\infty) >0$, that is, the
process exhibits transient behavior with positive probability, so
it does not converge in distribution to any probability measure on
$\stany$. In particular, $\|\pi_n - \pi\|_{\mathrm{TV}} \nrightarrow0$.
\end{prop}
%
%
\begin{remark}
In fact, we believe that in
Proposition~\ref{factXnonergodic},
$\mathbb{P}(X_{n, 1} \to\infty) = 1$,
though to reduce technicalities we only prove that
$\mathbb{P}(X_{n,1} \to\infty) >0$, which is sufficient
to establish nonergodicity.
\end{remark}

A detailed proof of Proposition~\ref{factXnonergodic} is
presented in Section~\ref{seccounterproof}. We also simulated
Example~\ref{exstairwaytoheaven2} on a computer [with the $(a_n)$ as
%
%
\begin{figure}

\includegraphics{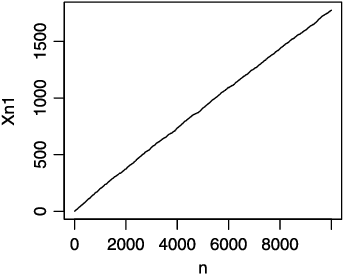}

\caption{Trace plot of $X_{n,1}$ from Example \protect\ref{exstairwaytoheaven2}.}
\label{fig1}
\end{figure}
defined in Section~\ref{seccounterproof}], resulting in the
trace plot of $X_{n,1}$ (Figure~\ref{fig1}) which illustrates the transient behavior since
$X_{n,1}$ increases quickly and steadily as a function of $n$:



\section{Ergodicity---the uniform case}
\label{secerguniform}


We now present positive results about ergodicity of adaptive
Gibbs samplers under various assumptions.
Results of this section are specific
to \textit{uniformly ergodic} chains. (Recall that a Markov chain with
transition kernel $P$ is uniformly ergodic if there exist $M<\infty$
and $\rho<1$ s.t. $ \|P^n(x, \cdot) - \pi(\cdot)\|_{\mathrm{TV}} \leq M\rho
^n$ for every $x \in\stany$; see, e.g.,
\cite{MeynTw,RobRos} for
this and other notions related to general state space Markov chains.)
In some sense this is a severe restriction, since most MCMC algorithms
arising in statistical applications are not uniformly ergodic. However,
truncating the variables involved at some (very large) value is usually
sufficient to ensure uniform ergodicity without affecting the statistical
conclusions in any practical sense, so the results of this section may
be sufficient for a pragmatic user. The nonuniform case is considered
in the following Section~\ref{secergnonunif}.

To continue, recall that \texttt{RSG}($\alpha$) stands for random
scan Gibbs
sampler with selection probabilities $\alpha$ as defined by
Algorithm~\ref{algrandomGibbs}, and \texttt{AdapRSG} is the adaptive
version as defined by Algorithm~\ref{algGibbsadap}.
For notation, let $\Delta_{d-1}:= \{(p_1,\ldots, p_d) \in
\mathbb{R}^{d}\dvtx p_i
\geq0, \sum_{i=1}^d p_i = 1\}$ be the $(d-1)$-dimensional
probability simplex, and let
%
%
\begin{equation} \label{Ydef}
\mathcal{Y}:= [\varepsilon, 1]^d \cap\Delta_{d-1}
\end{equation}
for some $0 < \varepsilon\leq1/d$.
We shall assume that all our selection probabilities are in
this set $\Y$.
%
%
\begin{remark} \label{remepsilon} The above assumption may seem
constraining, it is, however, irrelevant in practice. The additional
computational effort on top of the unknown optimal strategy $\alpha^*$
(that may be in $\Delta_{d-1} - \Y$) is easily controlled by setting
$\varepsilon:= (Kd)^{-1}$ that effectively upperbounds it by $1/K$.
The argument can be easily made rigorous, for example, in terms of the
total variation distance or the asymptotic variance.
\end{remark}

\subsection{Adaptive random scan Gibbs samplers}
\label{secadaprandscanGibbs}

The main result of this section is the following theorem.

%
%
\begin{theorem}\label{thmuniformmain}
Let the selection probabilities $\alpha_n \in\Y$ for all $n$,
with $\Y$ as in~(\ref{Ydef}).
Assume that:
\begin{longlist}[(b)]
\item[(a)] $|\alpha_n - \alpha_{n-1}| \to0$ in probability for
fixed starting values $x_0 \in\stany$ and $\alpha_0 \in\mathcal{Y}$.
\item[(b)] there exists $\beta\in\mathcal{Y}$ s.t.
\texttt{RSG}($\beta$) is uniformly ergodic.
\end{longlist}
Then \texttt{\textit{AdapRSG}} is ergodic, that is,
%
%
\begin{equation}\label{eqnthmunifmain1}
T(x_0, \alpha_0, n) \to 0
\qquad\mbox{as } n \to\infty.
\end{equation}
Moreover, if:
\begin{longlist}[(a$'$)]
\item[(a$'$)] $\sup_{x_0, \alpha_0}|\alpha_n - \alpha_{n-1}| \to0$
in probability,
\end{longlist}
then convergence of \texttt{\textit{AdapRSG}} is also uniform over all $x_0,
\alpha_0$, that is,
%
%
\begin{equation}\label{eqnthmunifmain2} \sup_{x_0, \alpha_0}T(x_0,
\alpha_0, n) \to0 \qquad\mbox{as } n \to\infty.
\end{equation}
\end{theorem}
%
%
\begin{remark}\label{remafterthmunif}
(1) Assumption (b) will typically be
verified for $\beta= (1/d,\ldots,\allowbreak 1/d)$; see also
Proposition~\ref{propcheckingunifergforG} below.

\begin{longlist}[(3)]
\item[(2)] We expect that most adaptive random scan Gibbs samplers will be
designed so that $|\alpha_n - \alpha_{n-1}| \leq a_n$ for every
$n\geq
1$, $x_0 \in\stany$, $\alpha_0\in\mathcal{Y}$, and $\omega\in
\Omega$,
for some deterministic sequence $a_n \to0$ (which holds, e.g., for the
adaptations considered in~\cite{sylvia2}).
In such cases, (a$'$) is automatically satisfied.

\item[(3)] The sequence $\alpha_n$ is not required to
converge and, in particular, the amount of adaptation, that is,
$\sum_{n=1}^{\infty}|\alpha_n-\alpha_{n-1}|$, is allowed to be infinite.

\item[(4)] In Example~\ref{exstairwaytoheaven2},
condition (a$'$) is satisfied but condition (b) is not.

\item[(5)] If we modify Example~\ref{exstairwaytoheaven2} by truncating
the state space to say $\tilde{\stany} = \stany\cap(\{1,\ldots,
M\}\times\{1,\ldots, M\})$ for some $1<M<\infty$, then the corresponding
adaptive Gibbs sampler is ergodic and (\ref{eqnthmunifmain2}) holds.
\end{longlist}
\end{remark}

Before we proceed with the proof of Theorem~\ref{thmuniformmain},
we need some preliminary lemmas, which may be of independent interest.
%
%
\begin{lemma}\label{lemmauniformlyuniformly}
Let $\beta\in\mathcal{Y}$
with $\Y$ as in (\ref{Ydef}).
If \texttt{RSG}$(\beta)$ is uniformly ergodic, then also
\texttt{RSG}$(\alpha)$ is uniformly ergodic for every $\alpha\in\mathcal
{Y}$. Moreover, there exist $M< \infty$ and $\rho< 1$ s.t. $ \sup
_{x_0 \in\stany, \alpha\in\mathcal{Y}}T(x_0, \alpha, n) \leq
M\rho^n
\to0$.
\end{lemma}
\begin{pf} Let $P_{\beta}$ be the transition kernel of
\texttt{RSG}($\beta$). It is well known that for uniformly ergodic Markov
chains the whole state space $\stany$ is small (cf. Theorems 5.2.1 and
5.2.4 in~\cite{MeynTw} with their $\psi=\pi$). Thus there exists $s>0$,
a probability measure $\mu$ on $(\stany, \borel)$ and a positive
integer $m$, s.t. for every $x \in\stany$,
%
%
\begin{equation} \label{eqnminorization}
P_{\beta}^m(x, \cdot) \geq s\mu(\cdot).
\end{equation}
Fix $\alpha\in\mathcal{Y}$ and let
\[
r:= \min_i \frac{\alpha_i}{\beta_i}.
\]
Since $\beta\in\mathcal{Y}$, we have $1 \geq r \geq\frac
{\varepsilon
}{1-(d-1)\varepsilon} > 0$ and $P_{\alpha}$ can be written as a mixture
of transition kernels of two random scan Gibbs samplers, namely,
\[
P_{\alpha} = rP_{\beta} + (1-r)P_{q} \qquad\mbox{where } q=
\frac{\alpha- r\beta}{1-r}.
\]
This, combined with (\ref{eqnminorization}), implies
%
%
\begin{eqnarray}\label{eqnminorizationmixtureforranodmGibbs}
P_{\alpha}^m(x, \cdot) &\geq& r^mP_{\beta}^m(x, \cdot) \geq
r^ms\mu(\cdot) \nonumber\\[-8pt]\\[-8pt]
&\geq&
\biggl(\frac{\varepsilon}{1-(d-1)\varepsilon}\biggr)^m s \mu(\cdot)
\qquad\mbox{for every } x \in\stany.\nonumber
\end{eqnarray}
By Theorem 8 of~\cite{RobRos}, condition (\ref
{eqnminorizationmixtureforranodmGibbs}) implies
%
%
\begin{equation}\label{eqngibbsunif}
\|P_{\alpha}^n(x, \cdot) - \pi
(\cdot)\|_{\mathrm{TV}} \leq\biggl(1- \biggl(\frac{\varepsilon
}{1-(d-1)\varepsilon
}\biggr)^m s \biggr)^{\lfloor n/m\rfloor} \qquad\mbox{for all }
x
\in\stany.\hspace*{-35pt}
\end{equation}
Since the right-hand side of (\ref{eqngibbsunif}) does not depend on
$\alpha$, the claim follows.\vadjust{\goodbreak}~%
\end{pf}
%
%
\begin{lemma}\label{lemmauniformlyLipschitz} Let $P_{\alpha}$ and
$P_{\alpha'}$ be random scan Gibbs samplers using selection
probabilities $\alpha, \alpha' \in\mathcal{Y}:=[\varepsilon,
1-(d-1)\varepsilon]^d$
for some $\varepsilon>0$. Then
%
%
\begin{equation}\label{eqnlemmadiminadapbyweights}
\|P_{\alpha}(x, \cdot) - P_{\alpha'}(x, \cdot)\|_{\mathrm{TV}} \leq\frac
{|\alpha
- \alpha'|}{\varepsilon+ |\alpha- \alpha'|} \leq\frac{|\alpha-
\alpha
'|}{\varepsilon}.
\end{equation}
\end{lemma}
\begin{pf}
Let\vspace*{1pt} $\delta:= |\alpha- \alpha'|$. Then $r:=\min_i\frac{\alpha
'_i}{\alpha_i} \geq\frac{\varepsilon}{\varepsilon+\max_i|\alpha
_i -
\alpha'_i|} \geq\frac{\varepsilon}{\varepsilon+\delta}$ and, reasoning
as in the proof of Lemma~\ref{lemmauniformlyuniformly}, we can write
$P_{\alpha'} = rP_{\alpha} + (1-r)P_q$ for some $q$ and compute
\begin{eqnarray*}
\|P_{\alpha}(x, \cdot) - P_{\alpha'}(x,
\cdot
)\|_{\mathrm{TV}} & = & \bigl\|\bigl(rP_{\alpha} + (1-r)P_{\alpha}\bigr) - \bigl(rP_{\alpha} +
(1-r)P_q\bigr)\bigr\|_{\mathrm{TV}} \\
& = & (1-r) \|P_{\alpha} - P_q\|_{\mathrm{TV}}
\leq\frac{\delta}{\varepsilon+\delta}
\end{eqnarray*}
as claimed.
\end{pf}
%
%
\begin{cor}
$P_{\alpha}(x, B)$ as a function of $\alpha$ on $\mathcal{Y}$ is
Lipschitz with Lipschitz constant $1/\varepsilon$ for every fixed set
$B \in\borel$.
\end{cor}
%
%
\begin{cor} \label{cordiminishadap}
If $|\alpha_n-\alpha_{n-1}| \to0$ in probability, then also
\[
\sup_{x\in\stany}\|P_{\alpha_n}(x, \cdot) - P_{\alpha_{n-1}}(x, \cdot)\|_{\mathrm{TV}}
\to0
\]
in probability.
\end{cor}
\begin{pf*}{Proof of Theorem~\ref{thmuniformmain}}
We conclude the result from Theorem 1 of~\cite{RobRosJAP} that requires
simultaneous uniform ergodicity and diminishing adaptation.
Simultaneous uniform ergodicity results from combining assumption (b)
and Lemma~\ref{lemmauniformlyuniformly}. Diminishing adaptation results
from assumption (a) with Corollary~\ref{cordiminishadap}. Moreover,
note that Lemma~\ref{lemmauniformlyuniformly} is uniform in $x_0$ and
$\alpha_0$ and (a$'$) yields uniformly diminishing adaptation again by
Corollary~\ref{cordiminishadap}. A look into the proof of Theorem 1 of
\cite{RobRosJAP} reveals that this suffices for the uniform part of
Theorem~\ref{thmuniformmain}.
\end{pf*}

Finally, we note that verifying uniform ergodicity of a random scan Gibbs
sampler, as required by assumption (b) of Theorem~\ref{thmuniformmain},
may not be straightforward. Such issues have been investigated in, for
example,
\cite{RobertsPolson}, and more recently in relation to the parametrization
of hierarchical models (see~\cite{PapaGarethparametr} and references
therein). In the following proposition, we show that to verify uniform
ergodicity of any random scan Gibbs sampler, it suffices to verify
uniform ergodicity of the corresponding systematic scan Gibbs sampler
(which updates the coordinates $1, 2, \ldots, d$ in sequence rather than
select coordinates randomly). See also Theorem 2 of~\cite{galinneath}
for a related result.

%
%
\begin{prop}\label{propcheckingunifergforG}
Let $\alpha\in\mathcal{Y}$ with $\Y$ as in (\ref{Ydef}).
If the systematic scan Gibbs sampler is uniformly ergodic, then so
is \texttt{RSG}$(\alpha)$.\vadjust{\goodbreak}
\end{prop}

\begin{pf} Let
\[
P = P_1P_2\cdots P_d
\]
be the transition kernel of the uniformly ergodic systematic scan Gibbs
sampler, where $P_i$ stands for the step that updates coordinate $i$.
By the minorization condition characterization, there exist $s>0$, a
probability measure $\mu$ on $(\stany, \borel)$ and a positive integer
$m$, s.t. for every $x \in\stany$,
\[
P^m(x, \cdot) \geq s\mu(\cdot).
\]
However, the probability that the random scan Gibbs sampler $P_{1/d}$
in its $md$ subsequent steps will update the coordinates in exactly the
same order is $(1/d)^{md}>0$. Therefore, the following minorization
condition holds for the random scan Gibbs sampler.
\[
P_{1/d}^{md}(x, \cdot) \geq(1/d)^{md}s\mu(\cdot).
\]
We conclude that \texttt{RSG}($1/d$) is uniformly ergodic
and then, by Lemma~\ref{lemmauniformlyuniformly}, it
follows that \texttt{RSG}($\alpha$) is uniformly ergodic for
any $\alpha\in\Y$.
\end{pf}

\subsection{Adaptive random scan Metropolis-within-Gibbs}
\label{secadapMwG}


In this section we consider random scan Metropolis-within-Gibbs
sampler algorithms (see also Section~\ref{secnonunifadapRSadapMwG} for
the nonuniform case). Thus, given $X_{n-1,-i}$, the $i$th coordinate
$X_{n-1, i}$ is updated by a draw $Y$ from the proposal
distribution\break
$Q_{X_{n-1,-i}}(X_{n-1, i}, \cdot)$ with the usual Metropolis acceptance
probability for the marginal stationary distribution $\pi(\cdot|
X_{n-1, -i})$. Here, we consider algorithm \texttt{AdapRSMwG}, where
the proposal distributions $Q_{X_{n-1,-i}}(X_{n-1, i}, \cdot)$ remain
fixed, but the selection probabilities $\alpha_i$ are adapted on the fly.
We shall prove ergodicity of such algorithms under some circumstances.
(The more general algorithm \texttt{AdapRSadapMwG} is then considered
in the following section.)

To continue, let $P_{x_{-i}}$ denote the resulting Metropolis transition
kernel for obtaining $X_{n, i}|X_{n-1, i}$ given $X_{n-1, -i} = x_{-i}$.
We shall require the following assumption.

%
%
\begin{ass}\label{assuMetropstepsuniformly}
For every $i \in\{1,\ldots, d\}$ the transition kernel $P_{x_{-i}}$ is
uniformly ergodic for every $x_{-i} \in\stany_{-i}$. Moreover, there
exist $s_i >0$ and an integer $m_i$ s.t. for every $x_{-i} \in\stany
_{-i}$ there exists a probability measure $\nu_{x_{-i}}$ on $(\stany_i,
\mathcal{B}(\stany_i))$, s.t.
\[
P_{x_{-i}}^{m_i}(x_i, \cdot) \geq s_i \nu_{x_{-i}}(\cdot)
\qquad\mbox{for every } x_i \in\stany_i.
\]
\end{ass}

We have the following counterpart of Theorem~\ref{thmuniformmain}.
%
%
\begin{theorem}\label{thmunifMwG}
Let $\alpha_n \in\mathcal{Y}$ for all $n$,
with $\Y$ as in (\ref{Ydef}).
Assume that:
\begin{longlist}[(c)]
\item[(a)] $|\alpha_n - \alpha_{n-1}| \to0$ in probability for
fixed starting values $x_0 \in\stany$ and $\alpha_0 \in\mathcal{Y}$.\vadjust{\goodbreak}
\item[(b)] there exists $\beta\in\mathcal{Y}$ s.t.
\texttt{RSG}($\beta$) is uniformly ergodic.
\item[(c)] Assumption~\ref{assuMetropstepsuniformly} holds.
\end{longlist}
Then \texttt{\textit{AdapRSMwG}} is ergodic, that is,
%
%
\begin{equation}\label{eqnthmunifMwG1}
T(x_0, \alpha_0, n) \to0
\qquad\mbox{as } n \to\infty.
\end{equation}
Moreover, if:
\begin{longlist}[(a$'$)]
\item[(a$'$)] $\sup_{x_0, \alpha_0}|\alpha_n - \alpha_{n-1}| \to0$
in probability,
\end{longlist}
then convergence of \texttt{\textit{AdapRSMwG}} is also uniform over all $x_0,
\alpha_0$, that is,
%
%
\begin{equation}
\label{eqnthmunifMwG2} \sup_{x_0, \alpha_0}T(x_0, \alpha_0, n) \to
0 \qquad\mbox{as } n \to\infty.
\end{equation}
\end{theorem}


%
%
\begin{remark}\label{remafterthmunifMwG}
Remarks~\ref{remafterthmunif}(1)--(3) still
apply. Also, Assumption~\ref{assuMetropstepsuniformly} can easily be
verified in some cases of interest, for example:
\begin{longlist}[(3)]
\item[(1)] Independence samplers are essentially uniformly ergodic if and
only if the candidate density is bounded below by a multiple of the
stationary density, that is, $q(\dd x) \geq s \pi(\dd x)$ for some $s>0$;
cf.~\cite{MengersenTweedie}.

\item[(2)] The Metropolis--Hastings algorithm
with continuous and positive proposal density $q(\cdot, \cdot)$ and
bounded target density $\pi$ is uniformly ergodic if the state space is
compact; cf.~\cite{MeynTw,RobRos}.
\end{longlist}
\end{remark}

To prove Theorem~\ref{thmunifMwG} we build on the approach of \cite
{RobRoshybridAAP}. In particular, recall the following notions of
reversibility and of strong uniform ergodicity.

%
\begin{defi}
We say that a transition kernel $P$ on $\stany$ is reversible with
respect to its stationary distribution $\pi$, if for any $A, B \in
\mathcal{B}(\mathcal{X})$
\[
\int_A P(x, B) \pi(\dd x) = \int_B P(y, A) \pi(\dd y).
\]
\end{defi}
%
%
\begin{defi}
We say that a transition kernel $P$ on $\stany$ with stationary
distribution $\pi$ is $(m, s)$-\textit{strongly uniformly ergodic}, if
for some $s> 0$ and positive integer $m$
\[
P^m(x, \cdot) \geq s \pi(\cdot) \qquad\mbox{for every } x
\in
\stany.
\]
Moreover, we will say that a family of Markov chains $\{P_{\gamma
}\}_{\gamma\in\Gamma}$ on $\stany$ with stationary distribution
$\pi$ is $(m,s)$-\textit{simultaneously strongly uniformly ergodic}, if
for some $s> 0$ and positive integer $m$
\[
P_{\gamma}^m(x, \cdot) \geq s \pi(\cdot) \qquad\mbox{for every }
x \in\stany\mbox{ and } \gamma\in\Gamma.
\]
\end{defi}

By Proposition 1 in~\cite{RobRoshybridAAP}, if a Markov chain is both
uniformly ergodic and reversible, then it is strongly uniformly
ergodic. The following lemma improves over this result by controlling
both involved parameters.
%
%
\begin{lemma}\label{lemsimultaneousstronguniferg}
Let $\mu$ be a probability measure on $\stany$, let $m$ be a positive
integer and let $s>0$. If a reversible transition kernel $P$ satisfies
the condition
\[
P^m(x, \cdot) \geq s\mu(\cdot) \qquad\mbox{for every } x
\in
\stany,
\]
then it is $((\lfloor\frac{\log(s/4)}{\log
(1-s)}
\rfloor+2)m, \frac{s^2}{8})$-strongly uniformly ergodic.
\end{lemma}
\begin{pf}
By Theorem 8 of~\cite{RobRos}, for every $A \in\borel$ we have
\[
\|P^n(x, A) - \pi(A)\|_{\mathrm{TV}} \leq(1-s)^{\lfloor n/m \rfloor}
\]
and, in particular,
%
%
\begin{equation}\label{eqnlemproofTVdiffsmall}
\|P^{km}(x, A) - \pi(A)\|_{\mathrm{TV}} \leq s/4 \qquad\mbox{for } k
\geq\frac{\log(s/4)}{\log(1-s)}.
\end{equation}
Since $\pi$ is stationary for $P$, we have $\pi(\cdot) \geq s\mu
(\cdot
)$ and thus an upper bound for the Radon--Nikodym derivative
%
%
\begin{equation} \label{eqnRNboundmupi} \dd\mu/\dd\pi\leq1/s.
\end{equation}
Moreover, by reversibility,
\[
\pi(\dd x) P^m(x, \dd y) = \pi(\dd y) P^m(y, \dd x) \geq\pi(\dd y)s
\mu(\dd
x)
\]
and consequently
%
%
\begin{equation} \label{eqnlowerforPm}
P^m(x, \dd y) \geq s\bigl(\mu
(\dd x)/\pi(\dd x)\bigr) \pi(\dd y).
\end{equation}
Now define
\[
A:= \{x \in\stany\dvtx\mu(\dd x)/\pi(\dd x) \geq1/2\}.
\]
Clearly $\mu(A^c) \leq1/2$. Therefore by (\ref{eqnRNboundmupi}) we
have
\[
1/2 \leq\mu(A) \leq(1/s)\pi(A)
\]
and hence, $\pi(A) \geq s/2$. Moreover (\ref{eqnlemproofTVdiffsmall})
yields
\[
P^{km}(x, A) \geq s/4 \qquad\mbox{for } k := \biggl\lfloor
\frac
{\log(s/4)}{\log(1-s)}\biggr\rfloor+1
\]
and with $k$ defined above by (\ref{eqnlowerforPm}), we have
\begin{eqnarray*}
P^{km+m}(x, \cdot) &=& \int_{\stany}
P^{km}(x, \dd z)P^m(z, \cdot) \geq\int_{A} P^{km}(x, \dd z)P^m(z,
\cdot)
\\ & \geq& \int_{A} P^{km}(x, \dd z)(s/2) \pi(\cdot) \geq(s^2/8)
\pi
(\cdot).
\end{eqnarray*}
The proof is complete.
\end{pf}

We will need the following generalization of Lemma
\ref{lemmauniformlyuniformly}.
%
%
\begin{lemma}\label{lemRSMwGuniform}
Let $\beta\in\mathcal{Y}$ with $\Y$ as in (\ref{Ydef}). If
\texttt{RSG}$(\beta)$ is uniformly ergodic then there exist $s'>0$ and a
positive integer $m'$ s.t. the family
\{\texttt{RSG}($\alpha$)\}$_{\alpha\in\mathcal{Y}}$ is $(m',
s')$-simultaneously strongly uniformly ergodic.\vadjust{\goodbreak}
\end{lemma}
\begin{pf} $P_{\beta}(x, \cdot)$ is uniformly ergodic and reversible,
therefore, by Proposition~1 in~\cite{RobRoshybridAAP}, it is
$(m,s_1)$-strongly uniformly ergodic for some $m$ and $s_1$. Therefore,
and arguing as in the proof of Lemma~\ref{lemmauniformlyuniformly}
[cf. (\ref{eqnminorizationmixtureforranodmGibbs})] there exist $s_2
\geq(\frac{\varepsilon}{1-(d-1)\varepsilon})^m, $ s.t. for
every $\alpha\in\mathcal{Y}$ and every $x \in\stany$
%
%
\begin{equation}
P_{\alpha}^m(x, \cdot) \geq s_2 P_{\beta}^m(x, \cdot) \geq s_1 s_2
\pi
(\cdot).
\end{equation}
Set $m'=m$ and $s' = s_1s_2$.
\end{pf}
\begin{pf*}{Proof of Theorem~\ref{thmunifMwG}}
We proceed as in the proof of Theorem~\ref{thmuniformmain}, that is,
establish diminishing adaptation and simultaneous uniform ergodicity
and conclude (\ref{eqnthmunifMwG1}) and (\ref{eqnthmunifMwG2})
from Theorem 1 of~\cite{RobRosJAP}. Observe that
Lemma~\ref{lemmauniformlyLipschitz} applies for random scan
Metropolis-within-Gibbs algorithms exactly the same way as for random
scan Gibbs samplers thus diminishing adaptation results from assumption
(a) and Corollary~\ref{cordiminishadap}. To establish simultaneous
uniform ergodicity, observe that, by
Assumption~\ref{assuMetropstepsuniformly} and
Lemma~\ref{lemsimultaneousstronguniferg}, the Metropolis transition
kernel for $i$th coordinate, that is, $P_{x_{-i}}$, has stationary
distribution $\pi(\cdot|x_{-i})$ and is $((\lfloor
\frac{\log(s_i/4)}{\log(1-s_i)}\rfloor+2)m_i,
\frac{s_i^2}{8})$-strongly uniformly ergodic. Moreover,\vspace*{1pt} by
Lemma~\ref{lemRSMwGuniform}, the family \texttt{RSG}($\alpha$),
$\alpha\in\mathcal{Y}$ is $(m', s')$-strongly uniformly ergodic,
therefore, by Theorem 2 of~\cite{RobRoshybridAAP}, the family of
random scan Metropolis-within-Gibbs samplers with selection
probabilities $\alpha\in\mathcal{Y}$, \texttt{RSMwG}($\alpha$), is
$(m_*, s_*)$-simultaneously strongly uniformly ergodic with $m_*$ and
$s_*$ given as in~\cite{RobRoshybridAAP}.
\end{pf*}

We close this section with the following alternative version of
Theorem~\ref{thmunifMwG}.

%
%
\begin{theorem}\label{thmunifMwGalternative}
Let $\alpha_n \in\mathcal{Y}$ for all $n$,
with $\Y$ as in (\ref{Ydef}).
Assume that:
\begin{longlist}[(a)]
\item[(a)] $|\alpha_n - \alpha_{n-1}| \to0$ in probability for
fixed starting values $x_0 \in\stany$ and $\alpha_0 \in\mathcal{Y}$.
\item[(b)] there exists $\beta\in\mathcal{Y}$ s.t.
\texttt{RSMwG}$(\beta)$ is uniformly ergodic.
\end{longlist}
Then \texttt{\textit{AdapRSMwG}} is ergodic, that is,
%
%
\begin{equation}\label{eqnthmunifMwGalternative1} T(x_0, \alpha_0, n)
\to0 \qquad\mbox{as } n \to\infty.
\end{equation}
Moreover, if:
\begin{longlist}[(a$'$)]
\item[(a$'$)] $\sup_{x_0, \alpha_0}|\alpha_n - \alpha_{n-1}| \to0$
in probability,
\end{longlist}
then convergence of \texttt{\textit{AdapRSMwG}} is also uniform over all $x_0,
\alpha_0$, that is,
%
%
\begin{equation}\label{eqnthmunifMwGalternative2} \sup_{x_0, \alpha
_0}T(x_0, \alpha_0, n) \to0 \qquad\mbox{as } n \to\infty.
\end{equation}
\end{theorem}
\begin{pf} Diminishing adaptation results from assumption (a) and
Corollary~\ref{cordiminishadap}. Simultaneous uniform ergodicity can be
established as in the proof of Lem\-ma~\ref{lemmauniformlyuniformly}. The
claim follows from Theorem 1 of
\cite{RobRosJAP}.\vadjust{\goodbreak}
\end{pf}
%
%
\begin{remark}
Whereas the statement of Theorem~\ref{thmunifMwGalternative} may be
useful in specific examples, typically condition (b), the uniform
ergodicity of a random scan Metropolis-within-Gibbs sampler, will be
not available and establishing it will involve conditions required by
Theorem~\ref{thmunifMwG}.
\end{remark}

\subsection{Adaptive random scan adaptive Metropolis-within-Gibbs}
\label{secadaptadapt}

In this section, and also later in Section \ref
{secnonunifadapRSadapMwG}, we consider the adaptive random scan adaptive
Metropolis-within-Gibbs algorithm \texttt{AdapRSadapMwG}, that
updates both selection probabilities of the Gibbs kernel and proposal
distributions of the Metropolis step. Thus, given $X_{n-1,-i}$, the
$i$th coordinate $X_{n-1, i}$ is updated by a draw $Y$ from a proposal
distribution $Q_{X_{n-1,-i}, \gamma_{n,i}}(X_{n-1, i}, \cdot)$ with
the usual acceptance probability.
This doubly-adaptive algorithm has been used by, for example,~\cite{sylvia2},
for an application in statistical genetics. As with adaptive
Metropolis algorithms, the adaption of the proposal distributions in
this setting is motivated by optimal scaling results for random walk
Metropolis algorithms \cite
{RobertsGelmanGilks,RobRosMALA,Bedardaap,Bedardbeyond,AtchRobRosMCMCMC,RobRosscaling2001,RobRos,RobRosex,Rosenthalproposal}.

Let $P_{x_{-i}, \gamma_{n,i}}$ denote the resulting
Metropolis transition kernel for obtaining $X_{n, i}|X_{n-1,
i}$ given $X_{n-1, -i} = x_{-i}$. We will prove ergodicity
of this generalized algorithm using tools from the previous
section. Assumption~\ref{assuMetropstepsuniformly} must be reformulated
accordingly, as follows.
%
%
\begin{ass}\label{assuadapMetropstepsuniformly}
For every $i \in\{1,\ldots, d\}$, $x_{-i} \in\stany_{-i}$ and
$\gamma
_i \in\Gamma_i$, the transition kernel $P_{x_{-i}, \gamma_i}$ is
uniformly ergodic. Moreover, there exist $s_i >0$ and an integer $m_i$
s.t. for every $x_{-i} \in\stany_{-i}$ and $\gamma_i \in\Gamma_i$
there exists a probability measure $\nu_{x_{-i}, \gamma_i}$ on
$(\stany_i, \mathcal{B}(\stany_i))$, s.t.
\[
P_{x_{-i}, \gamma_i}^{m_i}(x_i, \cdot) \geq s_i \nu_{x_{-i},
\gamma
_i}(\cdot) \qquad\mbox{for every } x_i \in\stany_i.
\]
\end{ass}

We have the following counterpart of Theorems~\ref{thmuniformmain} and
\ref{thmunifMwG}.

%
%
\begin{theorem}\label{thmunifaMwG}
Let $\alpha_n \in\mathcal{Y}$ for all $n$,
with $\Y$ as in (\ref{Ydef}).
Assume that:
\begin{longlist}[(a)]
\item[(a)] $|\alpha_n - \alpha_{n-1}| \to0$ in probability for
fixed starting values $x_0 \in\stany$, $\alpha_0 \in\mathcal{Y}$ and
$\gamma_0 \in\Gamma$.
\item[(b)] there exists $\beta\in\mathcal{Y}$ s.t.
\texttt{RSG}$(\beta)$ is uniformly ergodic.
\item[(c)] Assumption~\ref{assuadapMetropstepsuniformly} holds.
\item[(d)] The Metropolis-within-Gibbs kernels exhibit diminishing
adaptation, that is, for every $i \in\{1,\ldots, d\}$ the $\mathcal
{G}_{n+1}$ measurable random variable
\[
\sup_{x\in\mathcal{X}}\|P_{x_{-i}, \gamma_{n+1,i}}(x_i, \cdot) -
P_{x_{-i}, \gamma_{n,i}}(x_i, \cdot)\|_{\mathrm{TV}} \to0
\qquad\mbox{in
probability, as } n \to\infty
\]
for fixed starting values $x_0 \in\stany$, $\alpha_0 \in\mathcal{Y}$
and $\gamma_0$.
\end{longlist}
Then \texttt{\textit{AdapRSadapMwG}} is ergodic, that is,
%
%
\begin{equation}\label{eqnthmunifaMwG1} T(x_0, \alpha_0, n) \to
0 \qquad\mbox{as } n \to\infty.\vadjust{\goodbreak}
\end{equation}
Moreover, if:
\begin{longlist}[(a$'$)]
\item[(a$'$)] $\sup_{x_0, \alpha_0}|\alpha_n - \alpha_{n-1}| \to0$
in probability,
\item[(d$'$)] $\sup_{x_0, \alpha_0} \sup_{x\in\mathcal{X}}\|
P_{x_{-i}, \gamma_{n+1,i}}(x_i, \cdot) - P_{x_{-i}, \gamma
_{n,i}}(x_i, \cdot)\|_{\mathrm{TV}} \to0$ in probability,
\end{longlist}
then convergence of \texttt{\textit{AdapRSadapMwG}} is also uniform over all
$x_0, \alpha_0$, that is,
%
%
\begin{equation}\label{eqnthmunifaMwG2} \sup_{x_0, \alpha_0}T(x_0,
\alpha_0, n) \to0 \qquad\mbox{as } n \to\infty.
\end{equation}
\end{theorem}
%
%
\begin{remark}
Remarks~\ref{remafterthmunif}(1)--(3) still
apply and Remark~\ref{remafterthmunifMwG} applies for verifying
Assumption~\ref{assuadapMetropstepsuniformly}. Verifying condition
(d) is discussed after the proof.
\end{remark}
\begin{pf*}{Proof of Theorem~\ref{thmunifaMwG}} We again proceed by establishing diminishing
adaptation and
simultaneous uniform ergodicity and concluding the result from Theorem
1 of~\cite{RobRosJAP}. To establish simultaneous uniform ergodicity
we proceed as in the proof of Theorem~\ref{thmunifMwG}. Observe
that by Assumption~\ref{assuadapMetropstepsuniformly} and
Lemma~\ref{lemsimultaneousstronguniferg} every adaptive
Metropolis transition kernel for $i$th coordinate, that is,
$P_{x_{-i},
\gamma_i}$, has stationary distribution $\pi(\cdot|x_{-i})$ and is
$((\lfloor\frac{\log(s_i/4)}{\log(1-s_i)}
\rfloor
+2)m_i, \frac{s_i^2}{8})$-strongly\vspace*{1pt} uniformly
ergodic. Moreover, by Lemma~\ref{lemRSMwGuniform} the family
\texttt{RSG}($\alpha$), $\alpha\in\mathcal{Y}$, is $(m', s')$-strongly
uniformly ergodic, therefore, by Theorem 2 of~\cite{RobRoshybridAAP},
the family of random scan Metropolis-within-Gibbs samplers with selection
probabilities $\alpha\in\mathcal{Y}$ and proposals indexed by
$\gamma
\in\Gamma$, is $(m_*, s_*)$-simultaneously strongly uniformly ergodic
with $m_*$ and $s_*$ given as in~\cite{RobRoshybridAAP}.

For diminishing adaptation we write
\begin{eqnarray*}
&& \sup_{x\in\mathcal{X}}\|P_{\alpha_n, \gamma_n}(x, \cdot) -
P_{\alpha_{n-1}, \gamma_{n-1}}(x, \cdot) \|_{\mathrm{TV}}
\\
&&\qquad\leq \sup_{x\in\mathcal{X}}\|
P_{\alpha_n, \gamma_n}(x, \cdot) - P_{\alpha_{n-1}, \gamma_{n}}(x,
\cdot) \|_{\mathrm{TV}} \\
&&\qquad\quad{}
+ \sup_{x\in\mathcal{X}}\|P_{\alpha_{n-1}, \gamma_n}(x, \cdot
) -
P_{\alpha_{n-1}, \gamma_{n-1}}(x, \cdot) \|_{\mathrm{TV}}.
\end{eqnarray*}
The first term above converges to $0$ in probability by Corollary \ref
{cordiminishadap} and assumption~(a). The second term
\begin{eqnarray*}
&& \sup_{x\in\mathcal{X}}\|P_{\alpha_{n-1}, \gamma_n}(x, \cdot) -
P_{\alpha_{n-1}, \gamma_{n-1}}(x, \cdot) \|_{\mathrm{TV}}
\\
&&\qquad \leq\sum_{i=1}^d \alpha_{n-1, i}\sup_{x\in
\mathcal{X}}\|P_{x_{-i}, \gamma_{n+1,i}}(x_i, \cdot) - P_{x_{-i},
\gamma_{n,i}}(x_i, \cdot)\|_{\mathrm{TV}}
\end{eqnarray*}
converges to $0$ in probability as a mixture of terms that converge to
$0$ in probability.
\end{pf*}

The following lemma can be used to verify assumption (d) of
Theorem~\ref{thmunifaMwG} (see also Example~\ref{examplediminish}
below).\vadjust{\goodbreak}
%
%
\begin{lemma} \label{lemmadiminishQP}
Assume that the adaptive proposals exhibit diminishing adaptation, that
is, for every $i \in\{1,\ldots, d\}$ the $\mathcal{G}_{n+1}$
measurable random variable
\[
\sup_{x\in\mathcal{X}}\|Q_{x_{-i}, \gamma_{n+1,i}}(x_i, \cdot) -
Q_{x_{-i}, \gamma_{n,i}}(x_i, \cdot)\|_{\mathrm{TV}} \to0
\qquad\mbox{in
probability, as } n \to\infty
\]
for fixed starting values $x_0 \in\stany$ and $\alpha_0 \in\mathcal{Y}$.

Then any of the following conditions:
\begin{longlist}[(ii)]
\item[(i)] The Metropolis proposals have symmetric densities, that
is,
\[
q_{x_{-i}, \gamma_{n,i}}(x_i, y_i) = q_{x_{-i}, \gamma
_{n,i}}(y_i, x_i),
\]
\item[(ii)] $\stany_i$ is compact for every $i$, $\pi$ is
continuous, everywhere positive and bounded,
\end{longlist}
implies condition \textup{(d)} of Theorem~\ref{thmunifaMwG}.
\end{lemma}
\begin{pf} 
The first statement can be concluded from Proposition 12.3 of~\cite
{AndrieuMoulines}, however, to be self-contained, we provide the argument.
Let $P_1$, $P_2$ denote transition kernels and $Q_1$, $Q_2$ proposal
kernels of two generic Metropolis algorithms for sampling from $\pi$ on
arbitrary state space $\stany$. To see that (i) implies (d) we
check that
\[
\|P_1(x, \cdot)-P_2(x, \cdot)\|_{\mathrm{TV}} \leq2\|Q_1(x, \cdot)-Q_2(x,
\cdot)\|_{\mathrm{TV}}.
\]
Indeed, the acceptance probability
\[
\alpha(x,y)=\min\biggl\{1, \frac{\pi(y)}{\pi(x)}\biggr\} \in[0,1]
\]
does not depend on the proposal, and for any $x\in\stany$ and $A \in
\borel$, we compute
\begin{eqnarray*}
|P_1(x, A)-P_2(x, A)| & \leq& \biggl|\int_A \alpha(x,y)
\bigl(q_1(y)-q_2(y)\bigr) \,\dd y\biggr| \\
&&{} + \mathbb{I}_{\{x \in A\}} \biggl|\int_{\mathcal{X}}
\bigl(1-\alpha(x,y)\bigr)\bigl(q_1(y)-q_2(y)\bigr) \,\dd y\biggr|
\\
& \leq& 2\|Q_1(x, \cdot)-Q_2(x, \cdot)\|_{\mathrm{TV}}.
\end{eqnarray*}
For the second statement note that condition (ii) implies there
exists $K < \infty$, s.t. $\pi(y)/\pi(x) \leq K$ for every $x,y \in
\stany$. To conclude that (d) results from (ii) note that
%
%
\begin{equation}\label{eqnabsval}
|{\min}\{a,b\} -\min\{c,d\}| < |a-c| +
|b-d|
\end{equation}
and recall acceptance probabilities $\alpha_i(x,y)=\min\{1,
\frac
{\pi(y)q_i(y,x)}{\pi(x)q_i(x,y)}\}$. Indeed, for any $x\in
\stany$
and $A \in\borel$, using (\ref{eqnabsval}), we have
\begin{eqnarray*}
&&
|P_1(x, A)-P_2(x, A)| \\
&&\qquad \leq \biggl| \int_A \biggl( \min\biggl\{q_1(x,y),
\frac{\pi(y)}{\pi(x)} q_1(y,x)\biggr\}
- \min\biggl\{q_2(x,y), \frac{\pi(y)}{\pi(x)} q_2(y,x) \biggr\}
\biggr) \,\dd y \biggr| \\
&&\qquad\quad{}
+ \mathbb{I}_{\{x \in A\}} \biggl|\int_{\mathcal{X}} \bigl(
\bigl(1-\alpha
_1(x,y)\bigr)q_1(x,y)
-
\bigl(1-\alpha_2(x,y)\bigr)q_2(x,y)\bigr) \,\dd y\biggr| \\
&&\qquad \leq 4(K+1) \|Q_1(x, \cdot)-Q_2(x, \cdot)\|_{\mathrm{TV}}
\end{eqnarray*}
and the claim follows since a random scan Metropolis-within-Gibbs
sampler is a mixture of Metropolis samplers.
\end{pf}

We now provide an example to show that diminishing adaptation of proposals
as in Lemma~\ref{lemmadiminishQP} does not necessarily imply condition
(d) of Theorem~\ref{thmunifaMwG} so some additional assumption is
required, for example, (i) or (ii) of Lemma~\ref{lemmadiminishQP}.
%
%
\begin{example} \label{examplediminish}
Consider a sequence of Metropolis algorithms with transition kernels
$P_1, P_2,\ldots$ designed for sampling from $\pi(k) = p^k(1-p)$ on
$\stany= \{0,1,\ldots\}$. The transition kernel $P_n$ results from
using proposal kernel $Q_n$ and the standard acceptance rule, where
\[
Q_n(j,k) = q_n(k):= \cases{\displaystyle
p^k
\biggl(\frac{1}{1-p}-p^n+p^{2n}\biggr)^{-1} ,&\quad for $k \neq n$,\vspace*{2pt}\cr
\displaystyle
p^{2n} \biggl(\frac{1}{1-p}-p^n+p^{2n}\biggr)^{-1} ,&\quad for $k = n$.}
\]
Clearly,
\[
\sup_{j \in\stany} \| Q_{n+1}(j,\cdot) -
Q_{n}(j,\cdot)\|_{\mathrm{TV}} = q_{n+1}(n) - q_n(n) \to0.
\]
However,
\begin{eqnarray*}
\sup_{j \in\stany} \| P_{n+1}(j,\cdot) - P_{n}(j,\cdot)\|_{\mathrm{TV}}
&\geq&
P_{n+1}(n,0) - P_{n}(n,0) \\
&=&
\min\biggl\{q_{n+1}(0),
\frac{\pi(0)}{\pi(n)}q_{n+1}(n)\biggr\} \\
&&{}
- \min\biggl\{q_{n}(0), \frac{\pi(0)}{\pi(n)}q_{n}(n)\biggr\} \\
& = & q_{n+1}(0) - q_n(0)p^n \\
&\to&1-p \neq0.
\end{eqnarray*}
\end{example}

\section{Ergodicity---nonuniform case}
\label{secergnonunif} \label{secnonunifadapRSadapMwG}

$\!\!\!$In this section we consider the case~where nonadaptive kernels are not
necessary uniformly ergodic. We study~adaptive random scan Gibbs
adaptive Metropolis-within-Gibbs (\texttt{AdapRSadapMwG}) algorithms in
the nonuniform setting, with parameters $\alpha\in\mathcal{Y}$ and
$\gamma_i \in\Gamma_i, i=1,\ldots, d$, subject to adaptation. The
conclusions we draw apply immediately to adaptive random scan Gibbs
Metropolis-within-Gibbs (\texttt{AdapRSMwG}) algorithms by keeping the
parameters $\gamma_i$ fixed for the Metropolis-within-Gibbs steps.

We keep the assumption that selection probabilities are in $\Y$ defined
in (\ref{Ydef}), whereas the uniform ergodicity assumption will be
replaced by some natural regularity conditions on the target density.

Our strategy is to use the generic approach of~\cite{RobRosJAP} and to
verify the diminishing adaptation and the containment conditions. The
containment condition has been extensively studied in~\cite{BaiRobRos}
and it is essentially necessary for ergodicity of adaptive chains (see
Theorem 2 therein for the precise result). In particular, containment
is implied by simultaneous geometrical ergodicity for the adaptive
kernels. More precisely, we shall use the following result
of~\cite{BaiRobRos}.

%
\begin{theorem}[(Corollary 2 of~\cite{BaiRobRos})]\label{thmforcontainment}
Consider the family $\{P_{\gamma}\dvtx\gamma\in\Gamma\}$ of Markov
chains on $\mathcal{X} \subseteq\mathbb{R}^d$, satisfying the
following conditions:
\begin{longlist}[(ii)]
\item[(i)] for any compact set $C \in\mathcal{B}(\mathcal{X})$,
there exist some integer $m>0$, and real $\rho>0$, and a probability
measure $\nu_{\gamma}$ on $C$ s.t.
\[
P^m_{\gamma}(x, \cdot) \geq\rho\nu_{\gamma}(\cdot)
\qquad\mbox{for all } x \in C,
\]
\item[(ii)] there exists a function $V\dvtx\mathcal{X} \to(1,
\infty
)$, s.t. for any compact set $C \in\mathcal{B}(\mathcal{X})$, we have
$\sup_{x\in C}V(x) < \infty$, $ \pi(V) < \infty$, and
\[
\limsup_{|x| \to\infty} \sup_{\gamma\in\Gamma} \frac{P_{\gamma
}V(x)}{V(x)} < 1,
\]
\end{longlist}
then for any adaptive strategy using $\{P_{\gamma}\dvtx\gamma\in\Gamma
\}
$, containment holds.
\end{theorem}

Throughout this section we assume $\mathcal{X}_i = \mathbb{R}$ for
$i=1,\ldots, d$, and $\mathcal{X} = \mathbb{R}^d$ and let $\mu_k$ denote
the Lebsque measure on $\mathbb{R}^k$. By $\{e_1,\ldots, e_d\}$ denote
the coordinate unit vectors and let $|\cdot|$ be the Euclidean norm.

Our focus is on random walk Metropolis proposals with symmetric
densities for updating $X_i|X_{-i}$ denoted as $q_{i,\gamma_i}(\cdot)$,
$\gamma_i \in\Gamma_i$. We shall work in the following setting,
extensively studied for nonadaptive Metropolis-within-Gibbs algorithms
in~\cite{FortJAP} (see also~\cite{RobRoselectr,RobRoshybridAAP} for
related work and~\cite{JarnerHansen} for analysis of the random walk
Metropolis algorithm).
%

\begin{ass}\label{assucontdens}
The target distribution $\pi$ is absolutely continuous with respect to
$\mu_d$ with strictly positive and continuous density $\pi(\cdot)$ on
$\mathcal{X}$.
\end{ass}

%
\begin{ass}\label{assulocpos}
The family $\{q_{i, \gamma_i} \}_{1 \leq i \leq d; \gamma_i \in
\Gamma
_i}$ of symmetric proposal densities with respect to $\mu_1$
(one-dimensional Lebesgue measure) is such that there exist
constants $\eta_i > 0, \delta_i > 0$, for $i=1,\ldots, d$, s.t.
%
%
\begin{equation}
\inf_{|x|\leq\delta_i}q_{i, \gamma_i}(x) \geq\eta_i
\qquad\mbox{for every } 1\leq i \leq d
\quad\mbox{and}\quad \gamma_i
\in
\Gamma_i.
\end{equation}
\end{ass}
%
%
\begin{ass}\label{assuseq} There exist $0 < \delta< \Delta\leq
\infty
$, such that
%
%
\begin{equation} \label{eqnfirst} 
\xi:= \inf_{1\leq i \leq d, \gamma_i \in\Gamma_i} \int_{\delta
}^{\Delta} q_{i, \gamma_i}(y) \mu_1(dy) > 0
\end{equation}
and,
for any sequence $x = \{x^j\}$ with $\lim_{j\to\infty}|x^j| =
+\infty
$, there exists a subsequence $\tilde{x} = \{\tilde{x}^j\}$ s.t. for
some $i \in\{1,\ldots, d\}$ and all $y \in[\delta, \Delta]$,
%
%
\begin{equation}\label{eqn54eq} 
\lim_{j \to\infty}\frac{\pi(\tilde{x}^j)}{\pi(\tilde{x}^j -
\operatorname{sign}(\tilde
{x}^j_i)ye_i)} = 0 \quad\mbox{and}\quad \lim_{j \to\infty}\frac{\pi
(\tilde
{x}^j + \operatorname{sign}(\tilde{x}^j_i)ye_i)}{\pi(\tilde{x}^j)} =
0.\hspace*{-28pt}
\end{equation}
\end{ass}

Discussion of the seemingly involved~\ref{assuseq} and simple
criterions for checking it are given in~\cite{FortJAP}. It was shown
in~\cite{FortJAP} that under these assumptions nonadaptive random scan
Metropolis-within-Gibbs algorithms are geometrically ergodic for
subexponential densities. We establish ergodicity of the doubly
adaptive \texttt{AdapRSadapMwG} algorithm in the same setting.
%
%
\begin{theorem}\label{thmnonunifexp}
Let $\pi$ be a subexponential density and let the selection
probabilities $\alpha_n \in\mathcal{Y}$ for all $n$, with $\mathcal
{Y}$ as in (\ref{Ydef}). Moreover assume that:
\begin{longlist}[(a)]
\item[(a)] $|\alpha_n - \alpha_{n-1}| \to0$ in probability for
fixed starting values $x_0 \in\stany$ and $\alpha_0 \in\mathcal{Y}$,
$\gamma_i \in\Gamma_i$, $ i= 1,\ldots, d;$
\item[(b)] The Metropolis-within-Gibbs kernels exhibit diminishing
adaptation, that is, for every $i \in\{1,\ldots, d\}$ the $\mathcal
{G}_{n+1}$ measurable random variable
\[
\sup_{x\in\mathcal{X}}\|P_{x_{-i}, \gamma_{n+1,i}}(x_i, \cdot) -
P_{x_{-i}, \gamma_{n,i}}(x_i, \cdot)\|_{\mathrm{TV}} \to0
\qquad\mbox{in
probability, as } n \to\infty
\]
for fixed starting values $x_0 \in\stany$ and $\alpha_0 \in\mathcal
{Y}$, $\gamma_i \in\Gamma_i$, $ i= 1,\ldots, d;$
\item[(c)]
Assumptions~\ref{assucontdens},~\ref{assulocpos},
\ref{assuseq} hold.
\end{longlist}
Then \texttt{\textit{AdapRSadapMwG}} is ergodic, that is,
%
%
\begin{equation}\label{eqnthmnonunifexp} T(x_0, \alpha_0, \gamma_0, n)
\to0 \qquad\mbox{as } n \to\infty.
\end{equation}
\end{theorem}

Before proving this result we state its counterpart for densities that
are log-concave in the tails. This is another typical setting carefully
studied in the context of geometric ergodicity of nonadaptive chains
\cite{FortJAP,RobRoshybridAAP,MengersenTweedie} where Assumption
\ref{assuseq} is replaced by the following two conditions.
%
%
\begin{ass}\label{assulogconcave}
There exists an $\phi> 0$ and $\delta$ s.t. $1/\phi\leq\delta<
\Delta\leq\infty$ and, for any sequence $x:= \{x^j\}$ with $\lim_{j
\to\infty}|x^j| = +\infty$, there exists a subsequence $\tilde
{x}:=\{
\tilde{x}^j\}$ s.t. for some $i \in\{1,\ldots, d\}$ and for all $y
\in
[\delta, \Delta]$,
%
%
\begin{eqnarray}
&&\lim_{j \to\infty}\frac{\pi(\tilde{x}^j)}{\pi(\tilde{x}^j -
\operatorname{sign}(\tilde
{x}^j_i)ye_i)} \leq\exp\{-\phi y\} \quad\mbox{and}
\nonumber\\[-8pt]\\[-8pt]
&&\lim_{j \to
\infty}\frac{\pi(\tilde{x}^j + \operatorname{sign}(\tilde{x}^j_i)ye_i)}{\pi
(\tilde
{x}^j)} \leq\exp\{-\phi y \}. \nonumber
\end{eqnarray}
\end{ass}
%
%
\begin{ass}\label{assnonunif14}
\[
\inf_{1 \leq i \leq d, \gamma_i \in\Gamma_i } \int_{\delta
}^{\Delta}
yq_{i,\gamma_i}(y) \mu_1(dy) \geq\frac{8}{\varepsilon\phi(e-1)}.
\]
\end{ass}
%
%
\begin{remark}
As remarked in~\cite{FortJAP}, Assumption~\ref{assulogconcave}
generalizes the one-dimensional definition of log-concavity in the
tails and Assumption~\ref{assnonunif14} is easy to ensure, at least if
$\Delta= \infty$, by taking the proposal distribution to be a mixture
of an adaptive component and a uniform on $[-U,U]$ for $U$ large enough
or a mean zero Gaussian with large enough variance.
\end{remark}
%
%
\begin{theorem}\label{thmnonuniflogconcave}
Let the selection probabilities $\alpha_n \in\mathcal{Y}$ for all $n$,
with $\mathcal{Y}$ as in~(\ref{Ydef}). Moreover, assume that:
\begin{longlist}[(a)]
\item[(a)] $|\alpha_n - \alpha_{n-1}| \to0$ in probability for
fixed starting values $x_0 \in\stany$ and $\alpha_0 \in\mathcal{Y}$,
$\gamma_i \in\Gamma_i$, $ i= 1,\ldots, d;$
\item[(b)] The Metropolis-within-Gibbs kernels exhibit diminishing
adaptation, that is, for every $i \in\{1,\ldots, d\}$ the $\mathcal
{G}_{n+1}$ measurable random variable
\[
\sup_{x\in\mathcal{X}}\|P_{x_{-i}, \gamma_{n+1,i}}(x_i, \cdot) -
P_{x_{-i}, \gamma_{n,i}}(x_i, \cdot)\|_{\mathrm{TV}} \to0
\qquad\mbox{in
probability, as } n \to\infty
\]
for fixed starting values $x_0 \in\stany$ and $\alpha_0 \in\mathcal
{Y}$, $\gamma_i \in\Gamma_i$, $ i= 1,\ldots, d;$
\item[(c)] Assumptions~\ref{assucontdens},~\ref{assulocpos},
\ref{assulogconcave},~\ref{assnonunif14} hold.
\end{longlist}
Then \texttt{\textit{AdapRSadapMwG}} is ergodic, that is,
%
%
\begin{equation}\label{eqnthmnonunifexp}
T(x_0, \alpha_0, \gamma_0, n)
\to0 \qquad\mbox{as } n \to\infty.
\end{equation}
\end{theorem}

We now proceed to proofs.
\begin{pf*}{Proof of Theorem~\ref{thmnonunifexp}}
Ergodicity will follow from Theorem 2 of~\cite{RobRosJAP} by
establishing diminishing adaptation and containment condition.
Diminishing adaptation can be verified as in the proof of Theorem \ref
{thmunifaMwG}. Containment will result from Theorem~\ref{thmforcontainment}.

Recall that $P_{\alpha, \gamma}$ is the random scan
Metropolis-within-Gibbs kernel with selection probabilities $\alpha$
and proposals indexed by $\{\gamma_i\}_{1\leq i \leq d}$. To verify the
small set condition (i), observe that Assumptions~\ref{assucontdens}
and~\ref{assulocpos} imply that for every compact set $C$ and every
vector $\gamma_i \in\Gamma_i$, $ i \in1,\ldots, d$, we can find $m^*$
and $\rho^*$ independent of $\{\gamma_i\}$, and such that
$P^{m^*}_{1/d, \gamma}(x, \cdot) \geq\rho^* \nu(\cdot)$ for all $x
\in C$. Hence, arguing as in the proof of Lemma \ref
{lemmauniformlyuniformly}, there exist $m$ and $\rho$, independent of
$\alpha\in\mathcal{Y}$ and $\{\gamma_i\}$, such that $P^{m}_{\alpha,
\gamma}(x, \cdot) \geq\rho\nu(\cdot)$ for all $x \in C$.

To establish the drift condition (ii), let $V_s:= \pi(x)^{-s}$ for
some $s \in(0,1)$ to be specified later. Then by Proposition 3 of
\cite
{RobRoshybridAAP}, for all $1 \leq i \leq d$, $\gamma_i \in\Gamma_i$,
and $x \in\mathbb{R}^d$ we have
%
%
\begin{equation} \label{eqndriftnottoomuch}
P_{i,\gamma_i} V_s(x) \leq r(s) V_s(x)
\qquad\mbox{where } r(s): =
1+s(1-s)^{1/s - 1}.
\end{equation}
Since $r(s) \to1$ as $s \to0$, we can choose $s$ small enough, so
that
%
%
\begin{equation} \label{eqnrssmallenough}\label{eq}
r(s) < 1 + \frac{\varepsilon\xi}{1 - 2 \varepsilon\xi}.
\end{equation}

The rest of the argument follows the proof of Theorem 2 in \cite
{FortJAP}. We repeat most of it since we need to ensure it is
independent of $\alpha$ and $\gamma$. Assume by contradiction that
there exists an $\mathbb{R}^d$-valued sequence $\{x^j\}$ s.t.
\[
\limsup_{j\to\infty} \sup_{ \alpha\in\mathcal{Y}, \gamma_i \in
\Gamma
_i, 1 \leq i \leq d} P_{\alpha, \gamma} V_s(x^j) / V_s(x^j) \geq1.
\]
Then there exists a subsequence $\{\hat{x}^j\}$ such that
\[
\lim_{j\to\infty} \sup_{ \alpha\in\mathcal{Y}, \gamma_i \in
\Gamma
_i, 1 \leq i \leq d} P_{\alpha, \gamma} V_s(\hat{x}^j) / V_s(\hat{x}^j)
\geq1.
\]
Moreover, as shown in~\cite{FortJAP}, proof of Theorem 2, page 129,
there exists an integer $k \in\{1,\ldots, d\}$ and a further
subsequence $\{\tilde{x}^j\}$, such that
%
%
\begin{equation}\label{eqndriftworks}
\lim_{j\to\infty} \sup_{\gamma_k \in\Gamma_k} P_{k, \gamma_k}
V_s(\tilde{x}^j) / V_s(\tilde{x}^j) \leq r(s) - \bigl(2r(s) - 1\bigr) \xi.
\end{equation}
The contradiction follows from (\ref{eqndriftnottoomuch}), (\ref
{eqnrssmallenough}) and (\ref{eqndriftworks}), since
\begin{eqnarray*}
&&
\lim_{j \to\infty} \sup_{ \alpha\in\mathcal{Y}, \gamma_i \in
\Gamma_i, 1 \leq i \leq d} \frac{P_{\alpha, \gamma}V_s(\tilde
{x}^j)}{V_s(\tilde{x}^j)} \\
&&\qquad= \lim_{j \to\infty} \sup
_{\alpha
\in\mathcal{Y}}\sum_{i=1}^d \alpha_i \sup_{\gamma_i \in\Gamma
_i} \frac
{P_{i, \gamma_i}V_s(\tilde{x}^j)}{V_s(\tilde{x}^j)} \\
&&\qquad=
\lim_{j \to\infty} \sup_{\alpha\in\mathcal{Y}}\biggl(\alpha_k
\sup
_{\gamma_k \in\Gamma_k} P_{k, \gamma_k} V_s(\tilde{x}^j) /
V_s(\tilde
{x}^j) + \sum_{i\neq k} \alpha_i \sup_{\gamma_i \in\Gamma_i}
\frac
{P_{i, \gamma_i}V_s(\tilde{x}^j)}{V_s(\tilde{x}^j)}\biggr) \\
&&\qquad\leq \varepsilon\bigl(r(s) - \bigl(2r(s) - 1\bigr) \xi\bigr) +
(1-\varepsilon) r(s) < 1.
\end{eqnarray*}
\upqed\end{pf*}
\begin{pf*}{Proof of Theorem~\ref{thmnonuniflogconcave}}
The proof is identical to the proof of Theorem~\ref{thmnonunifexp} with
the only difference that now the drift condition (ii) of Theorem~\ref
{thmforcontainment} will be established under Assumptions \ref
{assulogconcave} and~\ref{assnonunif14}.

Establishing (ii) of Theorem~\ref{thmforcontainment} will follow
closely the proof of Theorem~3 in~\cite{FortJAP}. Let again $V_s:= \pi
(x)^{-s}$ for some $s \in(0,1)$ to be specified later and recall that
(\ref{eqndriftnottoomuch}) holds for all $1 \leq i \leq d$, $\gamma_i
\in\Gamma_i$, and $x \in\mathbb{R}^d$. Assume by contradiction that
there exists an $\mathbb{R}^d$-valued sequence $\{x^j\}$ s.t.
\[
\limsup_{j\to\infty} \sup_{ \alpha\in\mathcal{Y}, \gamma_i \in
\Gamma
_i, 1 \leq i \leq d} P_{\alpha, \gamma} V_s(x^j) / V_s(x^j) \geq1.
\]
Then there exists a subsequence $\{\hat{x}^j\}$ such that
\[
\lim_{j\to\infty} \sup_{ \alpha\in\mathcal{Y}, \gamma_i \in
\Gamma
_i, 1 \leq i \leq d} P_{\alpha, \gamma} V_s(\hat{x}^j) / V_s(\hat{x}^j)
\geq1.
\]
Moreover, as shown in~\cite{FortJAP}, proof of Theorem 3, page 137,
equation (15), there exists an integer $k \in\{1,\ldots, d\}$ and a
further subsequence $\{\tilde{x}^j\}$, such that
%
%
\begin{eqnarray}\label{eqndriftagain}
\lim_{j\to\infty}
P_{k, \gamma_k} V_s(\tilde{x}^j) / V_s(\tilde{x}^j) & \leq& r(s) -
\bigl(2r(s) - 1\bigr)\mathcal{J}_{\gamma_k} (0) +
\mathcal
{J}_{\gamma_k}(\phi s)\nonumber\\[-8pt]\\[-8pt]
&&{}  + \mathcal{J}_{\gamma_k}\bigl(\phi(1-s)\bigr) -
\mathcal
{J}_{\gamma_k}(\phi),\nonumber
\end{eqnarray}
where for $b > 0$,
\[
\mathcal{J}_{\gamma_k}(b) = \int_{\delta}^{\Delta} e^{-by}
q_{k, \gamma_k}(y) \mu_1(dy).
\]
Now from (\ref{eqndriftnottoomuch}) and (\ref{eqndriftagain}) compute
\begin{eqnarray*}
&&\lim_{j \to\infty} \sup_{ \alpha\in\mathcal{Y}, \gamma_i \in
\Gamma_i, 1 \leq i \leq d} \frac{P_{\alpha, \gamma}V_s(\tilde
{x}^j)}{V_s(\tilde{x}^j)} \\
&&\qquad= \lim_{j \to\infty} \sup
_{\alpha
\in\mathcal{Y}}\sum_{i=1}^d \alpha_i \sup_{\gamma_i \in\Gamma
_i} \frac
{P_{i, \gamma_i}V_s(\tilde{x}^j)}{V_s(\tilde{x}^j)} \\
&&\qquad =
\lim_{j \to\infty} \sup_{\alpha\in\mathcal{Y}}\biggl(\alpha_k
\sup
_{\gamma_k \in\Gamma_k} P_{k, \gamma_k} V_s(\tilde{x}^j) /
V_s(\tilde
{x}^j) + \sum_{i\neq k} \alpha_i \sup_{\gamma_i \in\Gamma_i}
\frac
{P_{i, \gamma_i}V_s(\tilde{x}^j)}{V_s(\tilde{x}^j)}\biggr) \\
&&\qquad \leq r(s) - \varepsilon\inf_{\gamma_k \in\Gamma_k}
\bigl(\bigl(2r(s) - 1\bigr)\mathcal{J}_{\gamma_k} (0) + \mathcal{J}_{\gamma
_k}(\phi s)
+ \mathcal{J}_{\gamma_k}\bigl(\phi(1-s)\bigr) - \mathcal{J}_{\gamma_k}(\phi
)\bigr)
\\
&&\qquad = \sup_{\gamma_k \in\Gamma_k} \bigl( r(s) - \varepsilon
\bigl(\bigl(2r(s) - 1\bigr)\mathcal{J}_{\gamma_k} (0) + \mathcal{J}_{\gamma
_k}(\phi s)
+ \mathcal{J}_{\gamma_k}\bigl(\phi(1-s)\bigr) - \mathcal{J}_{\gamma_k}(\phi
)
\bigr)\bigr) \\
&&\qquad =: \sup_{\gamma_k \in\Gamma_k} \mathcal{H}(\gamma_k,
\phi
, s).
\end{eqnarray*}
The result will follow if we can find such an $s$ that $\sup_{\gamma_k
\in\Gamma_k} \mathcal{H}(\gamma_k, \phi, s) < 1$. Note that
$\mathcal
{H}(\gamma_k, \phi, 0) = 1$ for every $\gamma_k \in\Gamma_k$ and the
function is differentiable. Therefore, it is enough to show that there
exist $\kappa_1 > 0$ and $\kappa_2 > 0$
such that
\[
{\partial\over\partial s} \mathcal{H}(\gamma_k, \phi, s) <
- \kappa_1 \qquad\mbox{for all } \gamma_k \in\Gamma_k
\mbox{ and } s \in(0,\kappa_2)
\]
and conclude (ii) with $V_s(x) = \pi^{-s}(x)$ and $s:= \kappa_2$.
To this end compute
\begin{eqnarray*}
{1 \over\varepsilon} \,{\partial\over\partial s} \mathcal{H}(\gamma
_k, \phi, s) & = &
\biggl( {1\over\varepsilon} - 2 \mathcal{J}_{\gamma_k}(0)\biggr)
\,{\partial
\over\partial s} r(s)
- \phi\int_{\delta}^{\Delta} y e^{-\phi s y} q_{\gamma_k}(y) \mu_1(dy)
\\
&&{} + \phi\int_{\delta}^{\Delta} y e^{-\phi(1-s) y}
q_{\gamma_k}(y) \mu_1(dy) \\
& = & {1 \over\varepsilon} {(1-s)^{1/s} \log(1-s) \over s(s-1)} -
\phi
I_1 + \phi I_2 =: \clubsuit,
\end{eqnarray*}
and notice that by $1/\phi\leq\delta$ and Assumption \ref
{assnonunif14}, for $s$ small enough we have
\begin{eqnarray*}
I_1 - I_2 & \geq& {e-1 \over2e } \int_{\delta}^{\Delta} y
q_{\gamma
_k}(y) \mu_1(dy) \\ & \geq& {e-1 \over2e } { 8 \over\varepsilon
\phi
(e-1) } = {4 \over\varepsilon\phi e}
\end{eqnarray*}
and
\[
{(1-s)^{1/s} \log(1-s) \over s(s-1)} \leq {2 \over e}.
\]
Consequently there exists $\kappa_2 > 0$ s.t. for all $s \in(0,
\kappa_2)$
\[
\clubsuit\leq{2 \over\varepsilon e} - {4 \phi\over\varepsilon
\phi e} = - {2 \over\varepsilon e} =: \kappa_1
< 0.
\]
\upqed\end{pf*}

%
%
\begin{example}
\label{exGLMM}
We now give an example involving a simple generalized linear mixed model.
Consider the model and prior given by
%
%
\begin{eqnarray}
\label{eqnglmmmodel}
Y_i &\sim& \operatorname{Pois} ( e^{\theta+ X_i} ), \\
X_i &\sim& N(0,1),\\
\theta&\sim& N(0,1).
\end{eqnarray}
The model is chosen to be extremely simple so as to not detract from the
argument used to demonstrate ergodicity of \texttt{adapRSadapMwG},
although this argument readily generalizes to different exponential
families, link functions and random effect distributions.

We consider simulating from the posterior distribution of $\theta,
\mathbf{X}$ given observations $y_1, \ldots, y_n$ using \texttt{adapRSadapMwG}.
More specifically
we set
%
%
\begin{equation} \label{eqnspecificq}
q_{x_{-i}, \gamma} (x_i, y_i) =
{{ \exp
\{-(y_i - x_i)^2/2\gamma\} }\over{
\sqrt{2 \pi\gamma} } },
\end{equation}
where the range of permissible scales $\gamma$ is restricted to be in some
range ${\Re} = [a,b]$ with $0<a\le b < \infty$. We are in the
subexponential tail case and specifically we have the following.
\end{example}
%
%
\begin{prop}
\label{propGLMM}
Consider \texttt{\textit{adapRSadapMwG}} applied to model (\ref
{eqnglmmmodel}) using
any adaptive scheme satisfying the conditions \textup{(a)} and \textup{(b)} of
Theorem~\ref{thmnonunifexp}. Then the scheme is ergodic.
\end{prop}

For the proof, we require the following definition from
\cite{FortJAP}. We let
\[
\Phi= \{\hbox{functions }\phi\dvtx\Rl^+ \to\Rl^+; \phi(x) \to
\infty
\hbox{ as }x \to\infty\}.
\]
\begin{pf*}{Proof of Proposition~\ref{propGLMM}}
According to Theorem~\ref{thmnonunifexp}, it remains to check conditions
\ref{assucontdens},~\ref{assulocpos},~\ref{assuseq} hold.
Conditions~\ref{assucontdens} and~\ref{assulocpos} hold by construction,
while condition~\ref{assuseq} consists of two separate conditions.
One of these, given in (\ref{eqnfirst}), holds by construction from
(\ref{eqnspecificq}).
Moreover,~\cite{FortJAP} shows that (\ref{eqn54eq}) can be replaced by
the following condition: there exist functions $\{ \phi_i \in\Phi,
1 \le i \le d\}$ such that
$i \in\{1,\ldots, d\}$ and
all $y \in[\delta, \Delta]$,
%
%
\begin{equation}\label{eqnalternative}
\lim_{|x_i | \to\infty}
\sup_{\{
x_{-i}; \phi_j(|x_j|) \le\phi_i (|x_i|), j \neq i
\}}
\frac{\pi(\tilde{x}^j)}{\pi(\tilde{x}^j - \operatorname{sign}(\tilde
{x}^j_i)ye_i)} = 0
\end{equation}
and
\begin{equation}
\lim_{|x_i| \to\infty}
\sup_{\{
x_{-i}; \phi_j(|x_j|) \le\phi_i (|x_i|), j \neq i
\}}
\frac{\pi(\tilde{x}^j +
\operatorname{sign}(\tilde{x}^j_i)ye_i)}{\pi(\tilde{x}^j)} = 0.
\end{equation}
Now take $\phi_i(x) = x$ for all $1\le i\le d$ so that (\ref{eqnalternative})
can be rewritten as the two conditions
%
%
\begin{eqnarray}\label{eqnalternative2}\quad
\lim_{|x_i | \to\infty}
\sup_{\{
x_{-i}; |x_j| \le|x_i|, j \neq i
\}}
\exp\biggl\{
\int_{-y}^0 \nabla_i \log\pi\bigl( x+\operatorname{sign}(x_i) z e_i\bigr) \,\dd z
\biggr\} &=& 0,\\
\label{plus}
\lim_{|x_i| \to\infty}
\sup_{
\{
x_{-i}; |x_j| \le|x_i|, j \neq i \}
}
\exp\biggl\{
\int_{0}^y \nabla_i \log\pi\bigl( x+\operatorname{sign}(x_i) z e_i\bigr) \,\dd z
\biggr\} &=& 0
\end{eqnarray}
for all $y \in[\delta, \Delta]$,
where $\nabla_i$ denotes the derivative in the $i$th direction.
We shall show that uniformly on the set
$S_i(x_i)$, which is defined to be $\{ x_{-i}; |x_j| \le|x_i|, j
\neq i \}$,
the function $\nabla_i \log\pi(x)$ converges to $-\infty$ as $x_i
\to+ \infty$ and to $ +\infty$ as $x_i$ approaches $-\infty$.

Now we have $d=n+1$ and let $i$ correspond to the component $x_i$ for
$1 \le i \le n$ with $n+1$ denoting the component $\theta$.
Therefore, for $1\le i \le n$,
\[
\nabla_i \log\pi(x) = -e^{\theta+ x_i} + y_i - x_i
\]
and
\[
\nabla_{n+1} \log\pi(x) = - \sum_{i=1}^n e^{\theta+ x_i} -
\sum_{i=1}^n y_i - \theta.
\]
Now for $x_i >0$, $1\le i \le n$
\[
\nabla_i \log\pi(x) \ge y_i -x_i,
\]
which is diverging to $-\infty$ independently of $x_{-i}$.
Similarly,
\[
\nabla_{n+1} \log\pi(x) \ge\sum_{i=1}^n y_i -\theta
\]
diverging to $-\infty$ independently of $\{x_i; 1 \le i \le n\}$.

For $x_i <0$, $1\le i \le n$ and $(x_{-i}, \theta) \in S_i (x_i)$,
\[
\nabla_i \log\pi(x) \le y_i -x_i +1
\]
again diverging to $+\infty$ uniformly.
Finally, for $\theta<0$ and $x \in S_{n+1} (\theta)$,
\[
\nabla_{n+1} \log\pi(x) \ge-n +\sum_{i=1}^n y_i - \theta,
\]
again demonstrating the required uniform convergence.
Thus ergodicity holds.
\end{pf*}
%
%
\begin{remark}
The random effect distribution in Example~\ref{exGLMM} can be altered to
give different results. For instance, if the distribution is doubly exponential,
Theorem~\ref{thmuniformmain} can be applied using very similar
arguments to those used above. Extensions to more complex hierarchical models
are clearly possible though we do not pursue this here.
\end{remark}
%
%
\begin{remark}
An important problem that we have not focused on involves the
construction of
explicit adaptive strategies. Since little is known about the
optimization of the
random scan random walk Metropolis, even in the nonadaptive case, this
is not
a straightforward question. We are engaged in further work exploring
adaptation to attempt to maximize a given optimality criterion for the
chosen class of samplers. Two possible strategies are:\looseness=-1
\begin{itemize}
\item
to scale the proposal variance to approach $2.4$ times the empirically
observed conditional variance;
\item
to scale the proposal variance to achieve an algorithm with acceptance
proportion approximately $0.44$.
\end{itemize}\looseness=0
Both these methods are founded in theoretical arguments (see, e.g.,
\cite{RobRosscaling2001}).
\end{remark}


\section{\texorpdfstring{Proof of Proposition \protect\ref{factXnonergodic}}{Proof of Proposition 3.2}}
\label{seccounterproof}

The analysis of Example~\ref{exstairwaytoheaven2} is somewhat
delicate since the process is both time and space inhomogeneous (as
are most nontrivial adaptive MCMC algorithms). To establish
Proposition~\ref{factXnonergodic}, we will define a couple of auxiliary
stochastic processes. Consider the following one-dimensional process
$\lancucht$ obtained from $\lancuch$ by
\[
\tilde{X}_n:= X_{n,1} +
X_{n,2}-2.
\]
Clearly $\tilde{X}_n - \tilde{X}_{n-1} \in\{-1,0,1\}$;
moreover, $X_{n,1}\to\infty$ and $X_{n,2} \to\infty$ if and only if
$\tilde{X}_n \to\infty$. Note that the dynamics of $\lancucht$ are
also both time and space inhomogeneous.\vadjust{\goodbreak}

We will also use an auxiliary random walk-like space homogeneous
process
\[
S_0 = 0 \quad\mbox{and}\quad S_n := \sum_{i=1}^n Y_i
\qquad\mbox{for } n\geq1,
\]
where $Y_1, Y_2,\ldots$ are independent random variables taking values
in $\{-1,0,1\}$. Let the distribution of $Y_n$ on $\{-1,0,1\}$ be
%
%
\begin{equation}\label{eqndistrofY} \nu_n:= \biggl\{\frac
{1}{4}-\frac{1}{a_n}, \frac{1}{2}, \frac{1}{4}+\frac{1}{a_n}
\biggr\}.
\end{equation}

We shall couple $\lancucht$ with $\lancuchs$, that is, define them on
the same probability space $\{\Omega, \mathcal{F}, \mathbb{P}\}$, by
specifying the joint distribution of $(\tilde{X}_n, S_n)_{n \geq0}$ so
that the marginal distributions remain unchanged. We describe the
details of the construction later. Now define
%
%
\begin{equation}\label{eqndefOmegaXS} \Omega_{\tilde{X}\geq S}:= \{
\omega\in\Omega\dvtx\tilde{X}_n(\omega) \geq S_n(\omega)
\mbox{ for
every } n\}
\end{equation}
and
%
%
\begin{equation}\label{eqndefOmegainfty} \Omega_{\infty}:= \{ \omega
\in\Omega\dvtx S_n(\omega) \to\infty\}.
\end{equation}
Clearly, if $\omega\in\Omega_{\tilde{X}\geq S}\cap\Omega_{\infty}$,
then $\tilde{X}_n(\omega) \to\infty$. In the sequel we show that for
our coupling construction
%
%
\begin{equation}\label{eqngoal} \mathbb{P}(\Omega_{\tilde{X}\geq
S}\cap
\Omega_{\infty}) > 0.
\end{equation}

We shall use Hoeffding's inequality for $S_k^{k+n}:= S_{k+n} - S_k$.
Since $Y_n \in[-1,1]$, it yields for every $t> 0$,
%
%
\begin{equation}
\label{eqnHoefffirsttime}
\mathbb{P}(S_k^{k+n} - \mathbb{E}S_k^{k+n} \leq-nt) \leq\exp\bigl\{
-\tfrac
{1}{2}nt^2\bigr\}.
\end{equation}
Note that $\mathbb{E}Y_n = 2/a_n$ and thus $\mathbb{E}S_{k}^{k+n} =
2\sum_{i=k+1}^{k+n}1/a_i$. The following choice for the sequence $a_n$
will facilitate further calculations. Let
\begin{eqnarray*}
b_0 & = & 0, \\
b_1 & = & 1000, \\
b_{n}& = &b_{n-1} \biggl(1 + \frac{1}{10 + \log(n)}\biggr)
\qquad\mbox{for } n \geq2,\\
c_n & = & \sum_{i=0}^{n}b_i,\\
a_n & = & 10+\log(k) \qquad\mbox{for }
c_{k-1} < n \leq c_{k}.
\end{eqnarray*}

%
\begin{remark} To keep notation reasonable we ignore the fact that
$b_n$ will not be an integer. It should be clear that this does not
affect proofs, as the constants we have defined, that is, $b_1$ and
$a_1$, are bigger than required.
\end{remark}
%
%
\begin{lemma}\label{lemmaSnabove} Let $Y_n$ and $S_n$ be as defined
above and let
%
%
\begin{eqnarray}\quad
\Omega_1 &:= & \{\omega\in\Omega\dvtx S_k =k \mbox{ for
every } 0 < k \leq c_1\},\\
\label{eqnlemmaSnabove}
\Omega_n &:=& \biggl\{\omega\in\Omega\dvtx S_k \geq\frac{b_{n-1}}{2}
\mbox{ for every } c_{n-1} < k \leq c_n\biggr\} \qquad\mbox{for
} n\geq2.
\end{eqnarray}
Then
%
%
\begin{equation} \label{eqnlemmaYnabove}
\mathbb{P}\Biggl( \bigcap_{n=1}^{\infty} \Omega_n\Biggr) > 0.
\end{equation}
\end{lemma}
%
%
\begin{remark}
Note that $b_n \nearrow\infty$ and therefore $\bigcap_{n=1}^{\infty}
\Omega_n \subset\Omega_{\infty}$.
\end{remark}
\begin{pf*}{Proof of Lemma~\ref{lemmaSnabove}}
With positive probability, say $p_{1,S}$, we have $Y_1 =\cdots=
Y_{1000}=1$ which gives $S_{c_1} = 1000 = b_1$. Hence, $\mathbb
{P}(\Omega_1) = p_{1,S} >0$. Moreover, recall that
$S_{c_{n-1}}^{c_{n}}$ is a sum of $b_n$ i.i.d. random variables with
$\mathbb{E}S_{c_{n-1}}^{c_{n}} = \frac{2b_n}{10+\log(n)}$. Therefore,
for every $n \geq1$ by Hoeffding's inequality with $t= 1/(10+\log
(n))$, we can also write
\[
\mathbb{P}\biggl(S_{c_{n-1}}^{c_n} \leq\frac{b_n}{10+\log(n)}\biggr)
\leq
\exp\biggl\{-\frac{1}{2}\frac{b_n}{(10+\log(n))^2}\biggr\}=:p_n.
\]
Therefore, using the above bound iteratively, we obtain
%
%
\begin{equation} \label{eqnSproof1}
\mathbb{P}(S_{c_1} = b_1, S_{c_n} \geq b_n
\mbox{ for every }
n \geq2) \geq p_{1,S}\prod_{n=2}^{\infty} (1-p_n).
\end{equation}
Note that $\{S_{c_n} \geq b_n\} \subseteq\Omega_n $ by the choice of
$b_n$, and hence, equation (\ref{eqnSproof1}) implies also
%
%
\begin{equation}
\mathbb{P}\Biggl( \bigcap_{n=1}^{\infty} \Omega_n\Biggr) \geq p_{1,S}
\prod_{n=2}^{\infty} (1-p_n).
\end{equation}
Clearly in this case
%
%
\begin{equation}\label{eqnsumpninfty} p_{1,S}\prod_{n=2}^{\infty
}(1-p_n) > 0 \quad\Leftrightarrow\quad \sum_{n=1}^{\infty} \log
(1-p_n) > -\infty\quad\Leftrightarrow\quad\sum_{n=1}^{\infty} p_n <
\infty.\hspace*{-32pt}
\end{equation}
We conclude (\ref{eqnsumpninfty}) by comparing $p_n$ with $1/n^2$. We
show that there exists $n_0$ such that for $n \geq n_0$ the series
$p_n$ decreases quicker than the series $1/n^2$ and therefore $p_n$ is
summable. We check that
%
%
\begin{equation}\label{eqncomparison}\log\frac{p_{n-1}}{p_{n}} >
\log
\frac{n^2}{(n-1)^2} \qquad\mbox{for } n \geq n_0.
\end{equation}
Indeed
\begin{eqnarray*}
\log\frac{p_{n-1}}{p_{n}} & = & -\frac{1}{2}\biggl(\frac
{b_{n-1}}{(10+\log(n-1))^2} - \frac{b_n}{(10+\log(n))^2}\biggr)\\
& = & \frac{b_{n-1}}{2}\biggl(\frac{11 + \log(n)}{(10+\log(n))^3} -
\frac{1}{(10+\log(n-1))^2} \biggr)\\
& = & \frac{b_{n-1}}{2}\biggl( \frac{(11 + \log(n))(10+\log(n-1))^2 -
(10+\log(n))^3}{(10+\log(n))^3(10+\log(n-1))^2}\biggr).
\end{eqnarray*}
Now recall that $b_{n-1}$ is an increasing sequence. Moreover, the
numerator can be rewritten as
\[
\bigl(10 + \log(n)\bigr)\bigl(\bigl(10+\log(n-1)\bigr)^2 - \bigl(10+\log(n)\bigr)^2\bigr)+ \bigl(10 +
\log(n-1)\bigr)^2;
\]
now use $a^2-b^2 = (a+b)(a-b)$ to identify the leading term $(10 + \log
(n-1))^2$. Consequently there exists a constant $C$ and $n_0 \in
\mathbb
{N}$ s.t. for $n \geq n_0$
\[
\log\frac{p_{n-1}}{p_{n}} \geq\frac{C}{(10+\log(n))^3} >
\frac{2}{n-1} > \log\frac{n^2}{(n-1)^2}.
\]
Hence, $ \sum_{n=1}^{\infty} p_n < \infty$ follows.
\end{pf*}

Now we will describe the coupling construction of $\lancucht$ and
$\lancuchs$. We already remarked that $\bigcap_{n=1}^{\infty} \Omega_n
\subset\Omega_{\infty}$. We will define a coupling that implies also
%
%
\begin{equation}\label{eqngoal1}
\mathbb{P}\Biggl(\Biggl(\bigcap
_{n=1}^{\infty} \Omega_n\Biggr) \cap\Omega_{\tilde{X}\geq S}
\Biggr) \geq
C \mathbb{P}\Biggl(\bigcap_{n=1}^{\infty} \Omega_n\Biggr)
\qquad\mbox{for some universal } C>0\hspace*{-35pt}
\end{equation}
and therefore
%
%
\begin{equation}\label{eqngoal2}\mathbb{P}(\Omega_{\tilde
{X}\geq
S}\cap\Omega_{\infty}) > 0.
\end{equation}
Thus nonergodicity of $\lancuch$ will follow from Lemma \ref
{lemmaSnabove}. We start with the following observation.
%
%
\begin{lemma}\label{lemmacouplingabove} There exists a coupling of
$\tilde{X}_{n} - \tilde{X}_{n-1}$ and $Y_n$, such that:
\begin{longlist}[(a)]
\item[(a)] For every $n \geq1$ and every value of $\tilde
{X}_{n-1}$
%
%
\begin{equation}
\mathbb{P}(\tilde{X}_{n} - \tilde{X}_{n-1} = 1, Y_n = 1) \geq
\mathbb
{P}(\tilde{X}_{n} - \tilde{X}_{n-1} = 1)\mathbb{P}(Y_n = 1).
\end{equation}
\item[(b)] Write even or odd $\tilde{X}_{n-1}$ as $\tilde{X}_{n-1}
= 2i-2$ or $\tilde{X}_{n-1} = 2i-3$, respectively. If $2i - 8 \geq
a_n$, then the following implications hold a.s.
%
%
\begin{eqnarray}\label{eqncouplingabove1}
Y_n = 1 & \quad\Rightarrow\quad&
\tilde{X}_{n} - \tilde{X}_{n-1} =1, \\
\label{eqncouplingabove2}
\tilde{X}_{n} - \tilde{X}_{n-1} =-1 & \quad\Rightarrow\quad& Y_n = -1.
\end{eqnarray}
\end{longlist}
\end{lemma}
\begin{pf} Property (a) is a simple fact for any two $\{
-1,0,1\}$ valued random variables $Z$ and $Z'$ with distributions say,
$\{d_1, d_2, d_3\}$ and $\{d_1', d_2', d_3'\}$. Assign $\mathbb
{P}(Z=Z'=1):= \min\{d_3, d_3'\}$ and (a) follows. To
establish (b) we analyze the dynamics of $\lancuch$ and
consequently, of $\lancucht$. Recall\vspace*{1pt} Algorithm~\ref{algGibbsadap} and
the update rule for $\alpha_n$ in (\ref{eqnformulaforalpha}). Given
$X_{n-1} = (i,j)$, the algorithm will obtain the value of $\alpha_n$ in
step (1); next draw a coordinate according to $(\alpha_{n, 1}, \alpha_{n,
2})$ in step (2). In steps (3) and (4) it will move according to conditional
distributions for updating the first or the second coordinate. These
distributions are
\[
(1/2, 1/2) \quad\mbox{and}\quad \biggl(\frac{i^2}{i^2+(i-1)^2},
\frac{(i-1)^2}{i^2+(i-1)^2}\biggr),
\]
respectively. Hence, given $X_{n-1} = (i,i)$, the distribution of
$X_{n} \in\{(i,\break i-1), (i,i), (i+1,i)\}$ is
%
%
\begin{eqnarray}\label{eqnXdynamics1}
&&\biggl( \biggl(\frac{1}{2}-\frac{4}{a_n}\biggr) \frac
{i^2}{i^2+(i-1)^2},\nonumber\\[-8pt]\\[-8pt]
&&\qquad 1-\biggl(\frac{1}{2}-\frac{4}{a_n}\biggr) \frac{i^2}{i^2+(i-1)^2}
-\biggl( \frac{1}{4} + \frac{2}{a_n} \biggr), \frac{1}{4} + \frac{2}{a_n}
\biggr),\nonumber
\end{eqnarray}
whereas if $X_{n-1} = (i,i-1)$, then $X_{n} \in\{(i-1, i-1), (i,i-1),
(i,i)\}$ with probabilities
%
%
\begin{eqnarray}\label{eqnXdynamics2}
&&\biggl(\frac{1}{4} - \frac{2}{a_n}, 1 - \biggl(\frac{1}{4} -
\frac
{2}{a_n} \biggr) - \biggl(\frac{1}{2}+\frac{4}{a_n}\biggr)\frac
{(i-1)^2}{i^2+(i-1)^2},\nonumber\\[-8pt]\\[-8pt]
&&\qquad\hspace*{106pt} \biggl(\frac{1}{2}+\frac{4}{a_n}
\biggr)\frac
{(i-1)^2}{i^2+(i-1)^2} \biggr),\nonumber
\end{eqnarray}
respectively. We can conclude the evolution of $\lancucht$. Namely, if
$\tilde{X}_{n-1} = 2i-2$, then the distribution of $\tilde{X}_n -
\tilde {X}_{n-1} \in\{-1,0,1\}$ is given\vspace*{1pt} by
(\ref{eqnXdynamics1}) and if $\tilde{X}_{n-1} = 2i-3$, then the
distribution of $\tilde{X}_n - \tilde {X}_{n-1} \in\{-1,0,1\}$ is given
by (\ref{eqnXdynamics2}). Let $\leq _{\mathrm{st}}$ denote stochastic
ordering. By simple algebra both measures defined in
(\ref{eqnXdynamics1}) and (\ref{eqnXdynamics2}) are stochastically
bigger than
%
%
\begin{equation}\label{eqnXdynamicsbound}
\mu_n^i = (\mu_{n,1}^i, \mu_{n,2}^i, \mu_{n,3}^i ),
\end{equation}
where
%
%
\begin{eqnarray}
\label{eqnintermeasureminus1}
\mu_{n,1}^i & = & \biggl(\frac{1}{4} - \frac{2}{a_n}\biggr)
\biggl(1+\frac
{2}{i}\biggr) = \frac{1}{4} - \frac{1}{a_n} - \frac{2i+8 -
a_n}{2ia_n},\hspace*{-35pt}\\
\mu_{n,2}^i & = & 1 -\biggl(\frac{1}{4} - \frac{2}{a_n}\biggr)
\biggl(1+\frac
{2}{i}\biggr) - \biggl(\frac{1}{4}+\frac{2}{a_n}\biggr)\biggl(1-\frac
{2}{\max\{
4,i\}}\biggr),\nonumber\hspace*{-35pt}\\
\label{eqnintermeasureplus1}
\mu_{n,3}^i &=& \biggl(\frac{1}{4}+\frac
{2}{a_n}\biggr)\biggl(1-\frac{2}{\max\{4,i\}}\biggr) = \frac
{1}{4} +
\frac{1}{a_n} + \frac{2\max\{4,i\}-8 - a_n}{2a_n \max\{4,i\}
}.\hspace*{-35pt}
\end{eqnarray}
Recall $\nu_n$, the distribution of $Y_n$ defined in (\ref
{eqndistrofY}). Examine (\ref{eqnintermeasureminus1}) and (\ref
{eqnintermeasureplus1}) to see that if $2i - 8 \geq a_n$, then $\mu_n^i
\geq_{\mathrm{st}} \nu_n$. Hence, in this case also, the distribution
of $\tilde{X}_{n} - \tilde{X}_{n-1}$ is stochastically bigger than the
distribution of $Y_n$. The joint probability distribution of $(\tilde
{X}_{n} - \tilde{X}_{n-1}, Y_n)$ satisfying (\ref{eqncouplingabove1})
and (\ref{eqncouplingabove2}) follows.
\end{pf}
\begin{pf*}{Proof of Proposition~\ref{factXnonergodic}}
Define
%
%
\begin{equation}
\Omega_{1, \tilde{X}}:= \{\omega\in\Omega\dvtx\tilde{X}_n -
\tilde
{X}_{n-1} =1 \mbox{ for every } 0 < n \leq c_1\}.
\end{equation}
Since the distribution of $\tilde{X}_n - \tilde{X}_{n-1}$ is
stochastically bigger than $\mu_n^i$ defined in (\ref
{eqnXdynamicsbound}) and $\mu_n^i(1) > c > 0$ for every $i$ and $n$,
\[
\mathbb{P}(\Omega_{1, \tilde{X}}) =: p_{1,\tilde{X}} > 0.
\]
By Lemma~\ref{lemmacouplingabove} (a) we have
%
%
\begin{equation}
\mathbb{P}(\Omega_{1, \tilde{X}}\cap\Omega_1 ) \geq
p_{1,S}
p_{1,\tilde{X}} > 0.
\end{equation}
Since $S_{c_1}=\tilde{X}_{c_1} = c_1 = b_1$, on $\Omega_{1, \tilde
{X}}\cap\Omega_1$, the requirements for Lemma \ref
{lemmacouplingabove}(b) hold for $n-1=c_1$. We shall use Lemma
\ref{lemmacouplingabove}(b) iteratively\vspace*{1pt} to keep
$\tilde{X}_n \geq S_n$ for every~$n$. Recall that we write $\tilde
{X}_{n-1}$ as $\tilde{X}_{n-1} = 2i-2$ or $\tilde{X}_{n-1} = 2i-3$. If
$2i - 8 \geq a_n$ and $\tilde{X}_{n-1} \geq S_{n-1}$, then by Lemma
\ref {lemmacouplingabove}(b) also $\tilde{X}_{n} \geq S_{n}$. Clearly
if $\tilde{X}_{k} \geq S_{k}$ and $S_k \geq\frac{b_{n-1}}{2}$ for
$c_{n-1} < k \leq c_n$ then $\tilde{X}_k \geq\frac{b_{n-1}}{2}$ for
$c_{n-1} < k \leq c_n$, hence,
\[
2i-2 \geq\frac{b_{n-1}}{2} \qquad\mbox{for } c_{n-1} < k \leq c_n.
\]
This in turn gives $2i - 8 \geq\frac{b_{n-1}}{2} - 6$ for $c_{n-1} < k
\leq c_n$ and since $a_k = 10 + \log(n)$, for the iterative
construction to hold, we need $b_{n} \geq32+2\log(n+1)$. By the
definition of $b_n$ and standard algebra we have
\[
b_n \geq1000\Biggl(1 + \sum_{i=2}^n \frac{1}{10+\log(n)}\Biggr)
\geq
32+2\log(n+1) \qquad\mbox{for every } n\geq1.
\]
Summarizing the above argument provides
\begin{eqnarray*}
\mathbb{P}(X_{n, 1} \to\infty) & \geq& \mathbb{P}(\Omega
_{\infty
} \cap\Omega_{\tilde{X}\geq S} ) \geq\mathbb
{P}
\Biggl(\Biggl(\bigcap_{n=1}^{\infty} \Omega_n\Biggr) \cap\Omega_{\tilde
{X}\geq
S}\Biggr) \\ &\geq& \mathbb{P}\Biggl(\Omega_{1, \tilde{X}} \cap
\Biggl(\bigcap_{n=1}^{\infty} \Omega_n\Biggr) \cap\Omega_{\tilde{X}\geq
S}\Biggr) \\ & \geq& p_{1,\tilde{X}}p_{1,S} \prod_{n=2}^{\infty}
(1-p_n) > 0.
\end{eqnarray*}
Hence, $\lancuch$ is not ergodic, and in particular, $\|\pi_n - \pi\|
_{\mathrm{TV}} \nrightarrow0$.
\end{pf*}


%

\printaddresses

\end{document}